\documentclass[10pt, conference,english, letterpaper,onecolumn]{IEEEtran}
\IEEEoverridecommandlockouts

\usepackage[cmex10]{amsmath}
\usepackage[mathscr]{eucal}
\usepackage{%
	amsfonts,%
	amssymb,%
	amsthm,%
	babel,%
	bbm,%
	cite,%
	dsfont,
	enumitem,%
	etoolbox,%
	float,%
	flafter,%
	graphicx,%
	lipsum,%
	mathtools,%
	multirow,%
	multicol,%
	pgf,%
	pgfplots,%
	subcaption,%
	tikz,%
	url,%
	verbatim,%
}
\usepackage[normalem]{ulem}
\usepgflibrary{shapes}
\usetikzlibrary{%
	arrows,%
	positioning,%
}
\usepackage[section]{placeins}

\theoremstyle{plain}
\newtheorem{thm}{Theorem}%[section]
\newtheorem{lem}[thm]{Lemma}
\newtheorem{prop}[thm]{Proposition}
\newtheorem{cor}[thm]{Corollary}

\theoremstyle{definition}

\newtheorem{appx}[thm]{Approximation}

\theoremstyle{remark}
\newtheorem{rem}{Remark}

\newcommand{\eqn}[1]{\begin{align}#1\end{align}}
\newcommand{\EQ}[1]{\begin{equation*}#1\end{equation*}}
\newcommand{\EQN}[1]{\begin{equation}#1\end{equation}}
\newcommand{\meq}[2]{\begin{xalignat*}{#1}#2\end{xalignat*}}

\newcommand{\ieeeproof}[1]{\begin{IEEEproof}#1\end{IEEEproof}}
\newcommand{\set}[1]{\left\{#1\right\}}

\newcommand{\SetIn}[1]{\mathds{1}_{\set{#1}}}

\newcommand{\R}{\mathbb{R}}

\newcommand{\E}{\mathbb{E}}

\newcommand{\Z}{\mathbb{Z}}

\newcommand{\bmu}{\bar{\mu}}
\newcommand{\bN}{\bar{N}}
\newcommand{\bY}{\bar{Y}}

\newcommand{\sY}{\mathscr{Y}}
\newcommand{\sZ}{\mathscr{Z}}

\newcommand{\bpi}{\ensuremath{\boldsymbol{\pi}}}

\renewcommand{\ge}{\geqslant}
\renewcommand{\le}{\leqslant}
\newcommand{\iid}{\emph{i.i.d.~}}

\newcommand{\nix}[1]{}

\nix{
	\usepackage[hidelinks,colorlinks=true,urlcolor=blue, %Colour for external hyperlinks
	linkcolor=blue, %Colour of internal links
	citecolor=red %Colour of citations
	]{hyperref}
}

\title{Modeling Performance and Energy trade-offs in Online Data-Intensive Applications}
\author{
Ajay Badita\IEEEauthorrefmark{1}%
\and Rooji Jinan\IEEEauthorrefmark{1}%
\and Balajee Vamanan\IEEEauthorrefmark{2}% 
\and Parimal Parag\IEEEauthorrefmark{1}
\thanks{
The authors\IEEEauthorrefmark{1} are with Indian Institute of Science, Bangalore, KA 560012, India. 
Email: \IEEEauthorrefmark{1}\{ajaybadita,roojijinan,parimal\}@iisc.ac.in.
}
\thanks{
Author\IEEEauthorrefmark{2} is with the department of CS, University of Illinois, Chicago, IL 60607, US. Email: \IEEEauthorrefmark{2}\{bvamanan\}@uic.edu.
}
}

\begin{document}
\maketitle
\begin{abstract}
We consider energy minimization for data-intensive applications run on large number of servers, 
for given performance guarantees.  
We consider a system, 
where each incoming application is sent to a set of servers, 
and is considered to be completed if a subset of them finish serving it. 
We consider a simple case when each server core has two speed levels,  
where the higher speed can be achieved by higher power for each core independently.
The core selects one of the two speeds probabilistically for each incoming application request. 
We model arrival of application requests by a Poisson process, 
and random service time at the server with independent exponential random variables.
Our model and analysis generalizes to today's state-of-the-art in CPU energy management where each core can independently select a speed level from a set of supported speeds and corresponding voltages.  
The performance metrics under consideration are the mean number of applications in the system and the average energy expenditure. 
We first provide a tight approximation to study this previously intractable problem {and derive closed form approximate expressions for the performance metrics when service times are exponentially distributed.} 
%under consideration}. 
Next, we study the trade-off between the approximate mean number of applications and energy expenditure in terms of the switching probability.
{We demonstrate that the numerically obtained curves are closely approximated by the expressions derived for the performance metrics. % when service times are exponentially distributed.
}
\end{abstract}
%\IEEEkeywords{ 
%FFork-join queues; Multi-core architecture; heterogeneous rates; Power consumption, Information retrieval}

%%%%%%%%%%%%%%%%%%%%%%%%%%%%%%%%%%%%%%%%%
\section{Introduction} 
\label{sec:intro}
%%%%%%%%%%%%%%%%%%%%%%%%%%%%%%%%%%%%%%%%% 
%\red{[What is an MDS coded system? The authors should at least give a brief introduction about the MSD coded system first.]}

In today's world of ubiquitous Internet connectivity and big data, 
we heavily rely on services such as Web search, social networks,
and e-commerce for curated access to vast amounts of data. 
Formally referred to as \textit{Online, Data-Intensive (OLDI)} applications, 
these applications are hosted in large datacenters, they consume 
hundreds of megawatts of power, and they incur millions of dollars in annual 
operating expenses~\cite{BarrosoSLCA2013, partha}. 
The problem of minimizing the energy consumption of 
OLDI applications in datacenters continues to garner significant interest 
in the computer systems research community~\cite{adrenaline, rubik, dirigent}.

Modern datacenter designs spend only about 10\% of their overall energy 
in non-computing needs such as cooling and power conversion, so  
most of the energy is spent directly in computing needs~\cite{google_pui}.
Among compute components, CPUs consume the largest 
fraction of the overall energy (see Figure 1.6 from~\cite{BarrosoSLCA2013}). 
Therefore, the focus of our work as well as prior works in the community 
is on achieving energy efficiency of CPUs in datacenters that host OLDI applications.

OLDI applications differ significantly from traditional  
applications in terms of their software architecture and performance needs. 
Unlike traditional applications, OLDI applications deal with extremely large datasets.
For instance, a Web search query must search the Web index, which has a size of 
about 1 exabyte~\cite{index}.
Owing to their large datasets, OLDI applications distribute data to a large 
number of servers and each query must lookup a large number of servers. 
Therefore, datacenters hosting OLDI applications \textit{cannot} 
shutdown a fraction of the servers, as is commonly done with traditional applications. 
Further, unlike traditional applications, OLDI applications are typically user-facing and 
require stringent service-level agreements on latency. 
Therefore, OLDI applications are \textit{not} amenable to batching.

Each OLDI query is sent to all servers hosting the relevant data, 
and the overall response is aggregated from results delivered by individual servers. 
The overall response time of an OLDI query has strict deadlines (e.g., 200 ms)~\cite{VamananACM2015}, 
and is sensitive to servers that finish last. 
Slow servers that delay the overall query response are called \emph{stragglers}. 
To prevent stragglers from degrading the response time, 
OLDI applications typically trade-off accuracy for speed by dropping responses from servers that arrive after the deadline. 
Nevertheless, to provide high accuracy, deadlines are set such that only a \textit{very small} fraction 
(e.g., 1\% or 0.1\%) of responses  are dropped.  
Therefore, the response time is sensitive to higher percentile completion times of individual servers (e.g., $99$-th percentile)~\cite{AlizadehACM2010, VamananACM2015}. 
%Such fine-grained, \textit{per query} partitioning of work in OLDI applications is in stark contrast to coarse-grained partitioning in traditional multi-processing applications (i.e., OLTP databases) where each query is processed by one server (or a small subset of servers). 
%Thus, in OLDI applications, several individual servers must co-operate in short timescales~\cite{wsc} and the overall performance is sensitive to the worst-case (e.g., $99^{\text{th}}$ percentile) behavior of sub-components.

%\red{Need to address the following weakness: The motivation is unconvinced why cloud service operators need to make a trade-off between performance and energy saving in OLDI applications. In general, to meet the tenant's SLAs, datacenter will firstly ensure its best performance. Compared with energy-saving, making a better performance is more important. As for making a trade-off between performance and energy-saving, it is also not new, lots of research have been proposed. What's more, achieving energy-saving involves many other hardware modules in the datacenter, but this paper merely concerns with CPU. In addition, the working scenarios concentrate on OLDI applications, but there are also many non-OLDI computation tasks for each server in practice. However, this paper doesn't mention the influence of these computations.} 

Instead of a strict deadline, we can also consider a threshold on the number of responses. 
An OLDI query can mitigate the impact of stragglers, 
by collecting responses from first $k$ out of $n$ servers to which the query was sent. 
We note that $100\times k/n$ indicates the percentile completion time of all servers. 
If the query accepts larger number of responses, then the results would be more accurate. 
However, this accuracy comes at the cost of increased overall response time. 
Thus, the threshold $k$ on the number of responses, controls the 
trade-off between speed and accuracy. 
%However, we do \textit{not} control this trade-off. 
The trade-off between speed and accuracy is determined by service providers 
based on the (predicted) revenue loss 
due to slower responses (i.e., $k$ is high) versus 
inaccurate results (i.e., $k$ is low).
Once set by providers, this trade-off becomes part of service-level agreement and 
operators (energy management) cannot violate it. 
Thus, our goal is to optimize energy savings while honoring this trade-off. 
We also remark that if the response time of individual servers are assumed to be independent and identically distributed (\emph{i.i.d.}), 
then the overall response time for a single query is $k$th order statistics of $n$ response times. 
When the number of servers $n$ is large, this can converge to a deterministic overall response time, 
and hence this threshold $k$ can be chosen to ensure the overall response time meets the strict deadline.

Throughout this work, we consider the case when all queries are sent to a fixed number of server cores $n$, 
and the query is considered completed when $k$ of them respond. 
Individual cores may be serving the previously received queries, 
and hence the incoming queries are queued at the cores. 
This is precisely the setting of $(n,k)$ fork-join queues~\cite{Dimakis2010TIT,Xiang2016TNET,Shah2016TCOM}, 
where each OLDI query is forked to $n$ cores, 
and $k$ of them are joined to get the aggregate response. 
Since the remaining requests corresponding to this query at $(n-k)$ cores are no longer needed, 
they are canceled immediately.

% 
%In order to improve response time of the system we introduce redundancy in the data storage and enable each query to be served by a large number of servers in parallel.
%An efficient way of introducing redundancy is by coding the stored data.
%Specifically, we suppose that the data to be accessed can be coded using a Maximum Distance Separable(MDS) $(N,K)$ code,  where the the queried data is encoded in to $N$ symbols stored on $N$ different servers and the query is assumed to be serviced once it is serviced by $K$ different servers.
%We assume that each query is forked in to all $N$ servers so that all available servers will service the queries in parallel.
%
%\red{
%In this model, a large number of servers are queried in parallel, 
%and the application is assumed to be serviced when a fraction of them respond. 
%This is akin to forking a queries in an MDS coded system, and joining a sufficient number of responses. 
%}

%OLDIs do not lend themselves well to conventional energy management, 
%which slows cores based on \emph{average} response times. 
We define the response time to be a high percentile of individual core completion times. 
That is, we consider the overall response time to be the $k$th response among all $n$ cores, 
where $k/n$ is large. 
In this work, we mathematically model and analyze the trade-offs between energy and 
performance (i.e., power and response time) in OLDI applications. 
In particular, we find the minimum energy requirement to meet a mean response time guarantee. 

We assume a finite number of processor power states (e.g., FID-VID pairs)~\cite{RotemMICRO2012}.
While we analyze for two processor states, our analysis can be generalized to a finite number of 
states. 
We consider the scenario, where each application gets served in one of the two rates depending on a probabilistic policy. 
That is, we suppose that we can choose to run each core at high or low speed randomly for each query independent of other cores.
{We observe that our proposed scheme requires the servers to switch the power state quite often.
However, our policy is practically feasible as systems researchers and industry
have greatly reduced the switching cost (i.e., overheads). 
For instance, Intel's Running Average Power Limit (RAPL) enables switching the power
state of each CPU core in less than 1 ms and is supported by modern CPUs from Intel.
}

It turns out that under our modeling assumptions, the system evolution can be modeled as a Markov process. 
We first observe that the Markov process under consideration can be thought of as a sequence of virtual queues in series, 
which pool their servers when idling. 
For a single process power state, this tandem queue has been studied in~\cite{Parag2017Infocom, Badita2019TIT}, 
where the system state is a sequence of the number of applications that have been served by $i \le k$ servers. 
Due to multiple processor power states and probabilistic slowdown, 
the resulting Markov process for our model is more complex. 
Specifically, the system state also includes the number of high rate servers at each virtual queue in the series. 
 
Computing the mean response time requires finding the equilibrium distribution of the Markov process under consideration. 
Finding the invariant distribution of this process is equivalent to finding the eigenvalues of a high dimensional operator, and as such it remains intractable even for single processor power state. 
We propose an independent \emph{tandem queue approximation} for two processor power states, 
and approximate the system evolution by a reversible Markov process for which the equilibrium distribution can be easily computed.  
We show that the mean sojourn time of an application in the approximating queue remains close 
to the actual empirical average sojourn time, 
for all range of arrival rates. 

While there are a number of proposals for OLDI energy management in the 
systems community, these proposals are empirical in nature 
(see Section~\ref{sec:related} for a discussion of related work). 
We are not aware of any theoretical studies that 
optimize energy given a performance (delay) constraint. 
%explore the 
%trade-off between energy and performance in these applications. 
Because studies \footnote{[n.d.],Speed Matters for Google Web Search,viewed 21 May 2021,https://services.google.com/fh/files/blogs/google\textunderscore delayexp.pdf, [n.d].,Velocity and the Bottom Line. http://radar.oreilly.com/2009/07/velocity-making-your-site-fast.html.}
%~\cite{user_exp1, user_exp2} 
have shown that 
even mild increases in latency in OLDIs can drastically hurt revenue, 
operators often employ conservative energy management strategies. 
We hope that providing theoretical grounding to energy management would 
enable operators to confidently deploy aggressive energy management policies
for OLDI applications without the fear of violating service-level agreements.

\subsection{Main Contributions}
Our study is a first step in the direction of analytical modeling of energy delay tradeoff in distributed information retrieval systems in the context of OLDI applications and we summarize our contributions below.
\begin{itemize}[noitemsep,topsep=0pt,parsep=0pt,partopsep=0pt]
\item We analytically model OLDI applications as $(n,k)$ fork-join queues and perform a systematic analysis of their power management considering that each processor core can work at two processor power states offering two different service rates.  
\item We introduce a probabilistic slowdown policy independent of the system state which enables easy implementation at all servers in a distributed fashion.
\item We show that any $(n,k)$ fork-join queues for two processor power states with probabilistic slowdown can be modeled as a series of virtual queues with pooled service. 
\item Owing to the complexity of analysing the actual system, we provide an approximating Markov process to find closed form expression for the limiting mean response time for an application and the mean power consumption in the system. 
\item Through numerical simulations, we validate that the theoretically derived approximate closed form expressions are very close to the actual values of the system performace metrics for all arrival rates when service times are exponentially distributed.  
\item Based on the approximate system, we characterize optimal slowdown probability to minimize energy consumption while satisfying a service guarantee on the mean response time.
\end{itemize}

%%%%%%%%%%%%%%%%%%%%%%%%%%%%%%%%%%%%%%%%%
\subsection{Background on OLDI applications} 
\label{sec:background}
%%%%%%%%%%%%%%%%%%%%%%%%%%%%%%%%%%%%%%%%% 
\begin{figure}[h]
	\centering
	\includegraphics[width=0.35\textwidth]{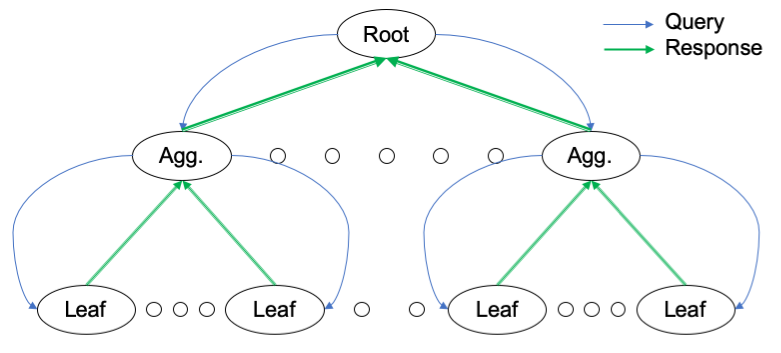}
	\caption{OLDI software architecture.}
	\label{fig:OLDI}
\end{figure}
%\figput{oldi1}{}{OLDI software architecture}
To facilitate and scale query-access among a large number of servers, 
OLDI applications have a hierarchical, \emph{partition-aggregate}, 
tree-like software architecture~\cite{AlizadehACM2010} (see Fig.~\ref{fig:OLDI}). 
The front-end server receives user queries from the Internet and 
sends queries to a root node (server). 
The root node, in turn, sends queries to a series of intermediate, aggregate nodes 
(denoted by "Agg." in Fig.~\ref{fig:OLDI}). 
Aggregate nodes perform two functions: 
(1) They forward queries to a subset of leaf nodes, which contain data to be searched. 
Subsequently, the leaf nodes process queries and retrieve relevant results.
(2) Once the leaf nodes return matching results, the aggregate nodes perform some intermediate processing (e.g., ranking the query results from a subset of servers) before returning results to the root node.  
Leaf nodes constitute a vast majority of servers (e.g., about 90\%)~\cite{VamananACM2015} and consume most of the power in OLDI datacenters because 
a large fraction of query processing occurs at leaf nodes.
Therefore, most existing work on OLDI power management~\cite{DavidACM2014, VamananACM2015, KastureISM2015} 
focus on slowing down leaf servers to save energy.  

%\red{Focus on a single OLDI query type with the subset of servers where it is routed to. 
%Explicit relation to $(n,k)$ fork-join queues.}

In such applications, each incoming query will be sent to all the leaf servers that hold the data relevant to the query. 
In the general case, a system would have many different types of queries, 
categorised into different classes. 
Accordingly, %we can model this system by allowing arrival of queries belonging to different classes where 
a query belonging to a certain class will be directed towards the subset of servers that hold data relevant to the particular class. 
In this work, we consider the simpler model of single class of queries, 
%by allowing only a single class of query and we 
and the subset of servers containing the relevant data. 
Therefore for single class case under consideration, we can assume that all servers have data relevant to each incoming query, without any loss of generality. 
Modern OLDI implementations avoid stragglers from affecting average query-response times
by dropping replies that arrive after a set \emph{deadline}.
The deadlines are carefully tuned to achieve a good trade-off between power and faster response times.  
OLDI applications are usually set based on some high percentile (e.g., $99^{\text{th}}$) completion times. 
In our model, 
we assume that the root node waits for $k$ responses out of $n$ leaf nodes. 
We have modeled the response time by a stochastic random variable, 
and hence we have implicitly assumed a stochastic deadline instead of a deterministic deadline. 
However, when the number of servers $n$ is large, 
this stochastic deadline can converge to a deterministic deadline. 
Once $k$ responses are received, query is dropped from the remaining $(n-k)$ leaf nodes. 
That is, each OLDI query must lookup data in all leaf servers as in a typical $(n,k)$ fork-join system. 
%%%%%%%%%%%%%%%%%%%%%%%%%%%%%%%%%%%%%%%%%
\subsection{Related Work}
\label{sec:related}
%%%%%%%%%%%%%%%%%%%%%%%%%%%%%%%%%%%%%%%%% 
A large body of prior work aggressively suggests powering down servers to save energy but suffer high latency, 
and therefore are not applicable for OLDIs~\cite{ChenACM2012, LangACM2010}. 
Unlike batch workloads, OLDIs require a global, cluster-wide approach for power management~\cite{MeisnerACM2011}.
Pegasus~\cite{DavidACM2014} and TimeTrader~\cite{VamananACM2015} address OLDI power management using 
centralized and distributed approaches, respectively. 
Several other papers address QoS and scheduling of OLDI jobs in datacenters~\cite{DelimitrouACM2013,CalinUSENIX2018,YangACM2013}.
Rubik~\cite{KastureISM2015} proposes a fine-grained DVFS scheme that controls the power states based on a statistical model without causing latency violations. 
Researchers have explored utilization of asymmetric multicore processors~\cite{HaqueISM2017}
and approximation~\cite{KulkarniCAL2018} to improve energy and performance. 
Other approaches include ensemble-level power budgeting~\cite{RanganathanSIGARCH2006}
and workload aware power budgeting~\cite{IsciISM2006}.
Network protocols for OLDIs also address tail latency for OLDIs~\cite{AlizadehACM2010,VamananSIGCOMM2012,WilsonSIGCOMM2011,AlmasiLANMAN2019}. 
However, unlike our work, all prior work is empirical. 
There is a large body of work on hard real time systems 
that deal with deadlines. 
However, unlike OLDIs, hard real time systems often operate 
under light load, their job completion times are predictable, 
they are centralized, and they are smaller in scale (e.g., a few racks). 
Therefore, those ideas are not directly applicable to OLDIs. 
There is existing work on co-locating batch (filler) jobs with latency-sensitive jobs 
like web search~\cite{CalinUSENIX2018}. Because OLDI jobs are prioritized ahead of 
filler jobs, which do not have latency constraints, our work is still 
applicable to derive optimal speeds for CPUs that host OLDI jobs.

%\red{Missing OLDI to fork-join connection. Perhaps background comes before related work?}

As discussed earlier, we model the query processing system as an $(n,k)$ fork-join system where each query is forked to all $n$ servers and is considered to be completed once $k$ out of $n$ responses are obtained. 
Analytically, $(n,k)$ fork-join queues for homogeneous servers with single rates have been studied previously in the literature~\cite{Dimakis2010TIT, Shah2016TCOM, Joshi2014JSAC, Xiang2014SIGMETRICS, Xiang2016TNET, Abbasi2018TNET, Badita2019TIT, Gardner2016QS, Wang2019QS, Abbasi2019TNSM}. 
Even for the simple case of servers operating at a single service rate with \iid exponential service times, 
the closed form expressions for mean sojourn time in $(n,k)$ fork-join queues is not available for general $n$. % in general. 
Derivation of stationary probabilities for a fork-join system with two heterogenous servers can be found in~\cite{Flatto198SIAM} when arrivals are poisson and service times are exponentially distributed.
The problem remains open for systems with larger number of servers.
For exponentially distributed service times and Poisson arrivals,  
bounds are presented in~\cite{Nelson1988TC, Baccelli1989JSTOR, Shah2016TCOM, Joshi2014JSAC, Xiang2014SIGMETRICS, Xiang2016TNET, Abbasi2018TNET}, 
and analytical approximations in~\cite{Badita2019TIT}. 
Exact analysis for special case of large systems is considered in~\cite{Li2016Infocom,Wang2019QS}, 
and small systems in~\cite{Gardner2016QS}.  
Special case of $(n,1)$ fork-join queues is studied in~\cite{Wang2014PER, Joshi2015Allerton,  Sun2016arXiv, Qiu2016ACM}, 
and $(n,n)$ fork-join queues in~\cite{Wang2015Sigmetrics,Wang2019TOMPECS}.
%,Aktas2018Sigmetrics}. 
Our analytical contribution differs from these previous works in the following ways.  
To the best of our knowledge, 
there has been no previous study of $(n,k)$ fork-join queues where all $n$ servers can operate at more than one rate.  
Further, all these works have a fixed service rate, and impact of server rate control on system performance has not been studied.

\emph{Notation:} 
We denote the set of reals by $\R$, 
the set of integers by $\Z$, 
the set of non-negative integers by $\Z_+$, 
the set of first $n$ positive integers by $[n]$. 

%%%%%%%%%%%%%%%%%%%%%%%%%%%%%%%%%%%%%%%%%
\section{System Model} 
\label{sec:model}
%%%%%%%%%%%%%%%%%%%%%%%%%%%%%%%%%%%%%%%%% 
We consider a large-scale cluster of $n \in \Z_+$ leaf servers, where $n$ is very large.  
We consider the case when the query processing time at each node is the bottleneck, 
and the communication time is comparatively negligible. 
This is the case when the database at each node is very large leading to large search times, 
and the nodes are connected by high bandwidth links leading to small communication times. 
    While the overall response time includes both compute (i.e., within CPU) 
    and communication delays (i.e., within network), the two components are 
    independent and can be optimized separately. 
    Our proposal considers only compute delays to optimize energy but if there is 
    slack from the network, we could further improve energy savings. 
    As such, accounting for network delays would provide more opportunity for 
    energy savings.
%     PICK EITHER RED or BLUE texts (below). 
%\blue{
%Moreover, even in applications communication delay cannot be neglected, query processing time contributes a significant fraction of the total response time.
%We note that the delay is due to the network elements is generally independent of the time spend at the processor 
%and in this work, we focus only on minimizing the query response time at the processor.
%}
Queries arriving in the system is modeled as a Poisson process with rate $\lambda$. 
We assume that each incoming query is sent to all $n$ leaf server nodes, where it is queued for the service. 
\subsection{Fork-join queueing system}
We assume a query response system where a query is assumed to be served if any $k$ of the $n$ leaf nodes reply to the root node. 
Each incoming query at the leaf nodes is queued for service. 
We refer to this query response system as an $(n,k)$ fork-join queueing system. 
In turn, the root node requests the straggling $n-k$ leaf nodes to drop this query from their respective queues, 
and this query leaves the system. 
We assume this feedback is error-free and instantaneous, and straggling lead nodes cancel this query immediately. 
It has been shown that cancellation delay for fork-join systems affect the system performance~\cite{LeeTIT2018}. 
However, we assume these delays to be negligible for analytical tractability, and optimistic performance gains. 
%We will see that the system remains highly intractable even under this idealized assumption.
Since we are looking at high percentile reception, we assume the fraction $\frac{k}{n}$ is very close to unity.  

%\red{[This paper assumes an instantaneous and error-free feedback. However, it is usually not true in practical scenarios. What if there exists an error or the job is abandoned. In addition, delay may exist due to a large volume of data traffics. The authors should consider these scenarios.]}

\subsection{Service}
The service time for each query at each leaf node is assumed to be random and independent accounting for unpredictable service variability due to multitude of independent background processes running on each server~\cite{ChengIMC2014}. 
Accordingly, the service time for $l$th query at leaf node $j$ is denoted by $T_{jl}$, 
and is assumed to be a random variable independent of all other random variables.  

The processing time distributions in compute clusters can be well modeled by shifted exponential distributions~\cite{LeeTIT2018, BitarISIT2017,XiangTNET2016,AlTNET2018}. %, Badita2020Infocom}. 
Shifted exponential service time is sum of a deterministic start-up time, and a random memoryless service time. 
When the initial constant shift is negligible compared to the mean of the memoryless part, 
the service distribution can be taken to be an exponential distribution. 
This fact together with the analytical tractability afforded by the memoryless distribution motivated us to take each service time to be distributed exponentially. 
Consideration of the fork-join queue is not work conserving for non-memoryless service distributions. 
Therefore, our assumption of memoryless service distribution captures the optimistic gains offered by the system. 

We assume that each leaf node selects a service rate from a finite set of available rates.
% corresponding to the heterogeneous cores.    
This assumption is motivated by support for multiple frequency--voltage levels for each core in today's multi-core CPUs~\cite{IsciISM2006}. 
For simplicity of presentation, we assume two service rates, i.e. $\E T_{jl} \in \set{\frac{1}{\mu_0}, \frac{1}{\mu_1}}$ where the rate $\mu_1 > \mu_0$ without any loss of generality. 
Our analysis can be generalized to any finite number of rates. 
  
A query is removed from a leaf node server queue if, either service is completed at this node, 
or $k$ other leaf nodes have completed service for this query.   
We assume that the server at each leaf node has a first-come-first-served (FCFS) queue.  
This assumption is motivated by analytical tractability, and the fact that the waiting times averaged over queries remain unchanged for any work-conserving queueing policy~\cite{NelsonTOC1988}. 
\subsection{Slowdown}
Each core is assumed to have two possible rates $\set{\mu_0, \mu_1}$ to choose from. 
For the best performance, each core should work at highest possible rate $\mu_1$. 
Therefore, any policy where certain cores are chosen to work at the lower rate $\mu_0$ is called a \emph{slowdown policy}. 
We consider a simple \emph{distributed probabilistic} slowdown policy, 
where each core \emph{independently} decides to speedup or slowdown at the beginning of each service.
We assume instantaneous speedup or slowdown for the ease of analytical tractability. 
Let $\xi_{j,l} \in \set{0,1}$ be a Bernoulli random variable where $\xi_{j,l}=1$ indicates that the $j$th leaf node is working on $l$th query at higher rate $\mu_1$. 
We assume $(\xi_{j,l}: j \in [n], l \ge 1)$ to be an \iid sequence with $\E \xi_{j,l} = p$, 
the probability of each forked query being served at high rate $\mu_1$.
As mentioned earlier, 
we adopt a probabilistic speed control that is distributed in nature. 
\subsection{Power Consumption}
It is well-established that power consumption of the server (CPU) is an increasing function $P(\cdot)$ of the service rate $\mu$~\cite{BorkarMICRO1999}. 
Therefore, we define $P_j \triangleq P(\mu_j)$ for $j \in \set{0,1}$. 
Since service rate $\mu_1 > \mu_0$, we have $P_1 > P_0$. 
A server is called \emph{active}, if it is working on some query.  
At time $t$, we denote the number of active servers with $M(t)$, 
and the number of active servers working at the higher service rate $\mu_1$ by $M_h(t)$. 
%and the number of active servers working at the lower service rate $\mu_0$ by $L(t) = N(t) - H(t)$. 
Their limiting means are given by 
%For any ergodic system, the limiting averages are equal to their limiting means, and hence 
\eqn{
\label{eqn:MeanNumServersLevel}
&\bar{M}_h \triangleq \lim_{T \to \infty}\frac{1}{T}\int_{0}^TM_h(t)dt,&
&\bar{M} \triangleq \lim_{T \to \infty}\frac{1}{T}\int_{0}^TM(t)dt.
}
Notice that $M(t) \le n$, since not all servers are working at all times.   
We assume that idling servers don't consume any power. 
Clearly, this is not true since there is some nominal power consumption 
even when a server is not working on any query. 
However, these additional consumptions remain even when a server is working on a query, 
and hence we ignore this constant power offset. 
%We also observe that the number of active servers working at lower service rate $\mu_0$ is given by $M_{\ell}(t) \triangleq M(t)-M_h(t)$. 
% 
Under the above assumptions, 
the power consumption at any time $t$, is given by $P(t) = P_0(M(t)-M_h(t)) + P_1M_h(t)$. 
The limiting average of power consumption is denoted by $\bar{P}$, 
and defined by 
\EQN{
\label{eqn:MeanPower}
\bar{P} \triangleq \lim_{T\to \infty}\frac{1}{T}\int_{0}^TP(t)dt 
= P_0\bar{M} + (P_1-P_0)\bar{M}_h. 
}

%For the slowdown policies of interest, $N(t)$ depends on time only through the system state, 
%and it's limiting mean is denoted by $\bar{N} \triangleq \lim_{t\to \infty}\frac{1}{t}\int_{0}^tN(s)ds$. 
%Hence, we can compute the limiting mean power consumption as 
%\EQ{
%\bar{P} %&\triangleq %\lim_{t\to \infty}\frac{1}{t}\int_{s=0}^tP(s)ds \\
%%&= \lim_{t\to \infty}\frac{1}{t}\int_{s=0}^t(P_0M_0(s)+ P_1M_1(s))ds\\
%=  \bar{N}P_0 + (P_1-P_0)\bar{M}_1. 
%}

\subsection{Performance Metrics} 
Let $R(t)$ be the number of queries in the system at any time $t$, 
then the limiting mean number of queries in the system is given by 
\EQN{
\label{eqn:MeanNumQueries}
\bar{R} \triangleq  \lim_{t \to \infty} \E[R(t)].
}
By Little's law, the limiting mean number of queries in the system is directly proportional to limiting mean sojourn time of the queries in the system, 
and is given by $\frac{1}{\lambda}\bar{R}$.  
Therefore, we focus on the limiting mean number of queries in the system $\bar{R} \in [0, \infty)$, 
as a measure of the processing performance of the system. 
Second measure of the performance is the limiting mean of the power consumption as defined in Eq.~\eqref{eqn:MeanPower}. 

%Since $P_0, \red{n} \violet{N}, P_1$ are constant, we would take the limiting mean number of servers $0 \le \bar{M}_1 \le N$ working at service rate $\mu_1$ as the measure of the power consumption. 

Our goal is to reduce both the limiting mean of the number of queries $\bar{R}$ and the limiting mean power consumption $\bar{P}$. 
It is clear at the outset that it is not possible to reduce both at the same time. 
For the case when $p=1$, 
all the servers work at the high rate $\mu_1$ or equivalently $M_h(t) = M(t)$ at all times $t$. 
In this case, the queries are processed at the fastest speed resulting in maximum power consumption.  
On the other hand when $p=0$, 
all the servers work at the low service rate $\mu_0$ or equivalently $M_h(t) = 0$ at all times $t$.  
In this case, the queries are processed at the slowest speed minimizing the power consumption. 
Naturally, there exists a trade-off between the mean query completion time and the power consumption at servers, 
depending on the number of active servers $M_h(t)$ working at the high service rate $\mu_1$ and the system load which affects the total number of active servers $M(t)$. 
This tradeoff is controlled by a single parameter $p$, probability of choosing high service rate $\mu_1$ for any server. 

%%%%%%%%%%%%%%%%%%%%%%%%%%%%%%%%%%%%%%%%%
\section{Background on fork-join queues}
\label{sec:backgroundForkJoin}
%%%%%%%%%%%%%%%%%%%%%%%%%%%%%%%%%%%%%%%%% 
We have assumed until now that each server consists of a single core which can work at two rates.
But in this section, we study the special case of single-rate servers. That is, the $(n,k)$ fork-join system has a single service rate $\mu_0 = \mu_1$, 
and the stochastic evolution of this system is Markov under our modeling assumptions. 
For the general case of multi-rate servers, the system state has to be embellished for the system evolution to remain Markov. 
Nevertheless, we use the system state for single-rate server as the stepping stone to the more complex general case. 
%Even for this special case, computation of the stationary distribution of the resulting Markov process remains intractable.  
%We therefore present an approximation introduced in~\cite{Parag2017Infocom} for single core servers, 
%that we will extend to our general case in the subsequent sections. 

Recall that we have a system of $n$ servers which are working at one of the two rates, 
namely, $\set{\mu_0, \mu_1}$ and each incoming query is forked in to all the $n$ servers. 
Queries are served at each server in FCFS fashion, and leave the server after response completion. 
Once $k$ servers have responded for a unique query, 
it immediately exits all the remaining $n-k$ servers before their response completions. 
Let $Y_i(t) \in \Z_+$ denote the number of queries at time $t$ with completed responses from $i$ servers. 
Since a query leaves the system after $k$ response completions, $i \in \set{0, \dots, k-1}$,  
and denote the vector $Y(t) \triangleq (Y_0(t), \dots, Y_{k-1}(t))$ that captures the status of queries in the  system at any time $t$. 
Let $N_i(t) \in \set{0, \dots, k-1}$ denote the number of active servers in the system available to serve a query with $i$ response completions if such a query exists in the system at time $t$. 
We define the sequence $N(t) \triangleq (N_0(t), \dots, N_{k-1}(t))$. 
\begin{rem}
Due to the FCFS service policy, 
the queries with $(i+1)$ responses are older than queries with $i$ responses,  
and the set of servers responding to queries with $i$ responses have already responded to the queries with $(i+1)$ responses. 
Therefore, if the head query of a server has $i$ responses, 
it has already served queries with $(i+1)$ responses, 
and is not useful to them. 
\end{rem}

\begin{rem}
An incoming query has zero responses, 
and when a query with $i$ responses gets served by one of the $N_i(t)$ servers,  
it becomes a query with $i+1$ completed responses. 
Hence, the sequence $Y(t) \triangleq (Y_0(t), \dots, Y_{k-1}(t))$ is akin to occupancy of $k$ queues in tandem, 
with incoming queries joining queue~$0$, and completed queries leaving the queue~$k-1$. 
We observe that a state transition from the current state $y \in \sY \triangleq \Z_+^{\set{0, \dots, k-1}}$ is a departure from queue $i-1$ that leads to an arrival to queue $i$. 
Defining $e_i$ as the unit vector in the $i$th dimension, 
the next state is denoted by 
\EQ{
a_i(y) \triangleq \begin{cases}
y + e_0, & i = 0,\\
y-e_{i-1}+e_i, & i \in [k-1],\\
y-e_{k-1},&i = k.
\end{cases}
}
\end{rem}
Next remark discusses a result from \cite[Corollary 4]{Badita2019TIT} which shows that the sequence $N(t)$ of number of active servers can be completely determined by the tandem queue occupancy $Y(t)$. 
\begin{rem}
\label{rem:N-Y_Relation}
For the $(n,k)$-fork-join queueing system described above, 
the vector $Y(t)$ represents the occupancies in a $k$-tandem queue at time $t$.  
If $Y(t) = y$, then the number of available servers to serve the head query at $i$th tandem queue is 
$N_i(t) = N_i(y)$ such that 
\EQN{
\label{eqn:NumServersPerQueue}
N_i(y) = 
\begin{cases}
n-k+1, & i = k-1,\\
1 + N_{i+1}(y)\SetIn{y_{i+1} = 0}, &i < k-1.
\end{cases}
}
\end{rem}
That is, when a query has $k-1$ response completions, it will be ahead of all other queries in all the remaining $n-(k-1)$ servers that are useful to it.
Similarly, a query with $i$ response completions will have $n - i$ servers useful to it for $i < k-1$. 
But all these servers except one will be useful to the queries with $i+1$ response completions as well.
So, in order for the query with $i$ response completions to get served by those servers, 
it has to wait until the queries with $i+1$ response completions complete its service at these servers.
That is, $N_i(t) = 1$ when there are queries with $i+1$ response completions in the system.  
If the system is devoid of any queries with $i+1$ response completions, then besides the already available single server, the query with $i$ response completions can be served by the servers available to serve the query with $i+1$ response completions. 

\begin{prop}\cite[Corollary 4]{Badita2019TIT}
\label{prop:CTMCNoSlowdown}
For the $(n,k)$ fork-join query response system under consideration, 
if the two service rates are identical, i.e. $\mu_0 = \mu_1 =\mu$, 
then $Y \triangleq (Y(t) \in \sY: t\in \R_+)$ is a continuous time Markov chain.
\end{prop}
\begin{proof}
We observed that the state of the system can be modeled by the tandem queue occupancy $Y(t)$, 
where the occupancy of the $i$th queue in tandem is given by the number of queries with $i$ completed responses $Y_i(t)$.  
As $Y_i(t) \in \Z_+$, for $i \in \{0,1,\cdots,k-1\}$, the state space is countable. 
Since the arrival and service completions are independent and memoryless, 
the inter-transition times remain memoryless, 
and independent of the past conditioned on the current sate $Y(t) = y$.  
From Remark~\ref{rem:N-Y_Relation}, we see that the total number of servers $N_i(t)$ serving the stage~$i$ at time $t$ is completely determined by the vector $Y(t)$. 
Therefore, when the two service rates are identical, 
the next inter-transition time for transition $a_i(y)$ at level~$i$, 
is given by the minimum of remaining service times for all active servers $N_{i-1}(y)$. 
Thus, it follows that the inter-transition times possess memoryless property and the jump probabilities depend only on the previous state. 
This proves that the process $Y$ is a CTMC~\cite[Chapter 4]{Ross1995Wiley}.
In addition, we can write the transition rate for the transition to the possible next state $a_i(y)$ as 
\EQ{
Q(y,a_i(y)) = \begin{cases}
\lambda,& i = 0,\\
{N_{i-1}(t)\mu} \SetIn{y_{i-1} > 0}, & i \in [k].
\end{cases}
}
\end{proof}

%%%%%%%%%%%%%%%%%%%%%%%%%%%%%%%%%%%%%%%%%
\section{System evolution}
\label{sec:Analysis} 
%%%%%%%%%%%%%%%%%%%%%%%%%%%%%%%%%%%%%%%%% 
In this section, we will study the evolution of an $(n,k)$ fork-join query response system, 
with distributed probabilistic slowdown for the case of single-core servers with rates in $\set{\mu_0,  \mu_1}$.
From Proposition~\ref{prop:CTMCNoSlowdown}, 
we observe that the system evolution can be modeled by the Markov process $Y$ that captures the tandem queue occupancy, 
when the two service rates are equal. 
In the general case, when the two rates are not equal, 
the inter-transition times remain memoryless. 
However, 
the transition rates are governed by the number of high and low rate servers at each stage of the tandem queue. 
Therefore, we define the number of high rate servers at stage~$i$ at time~$t$ by $H_i(t)$, 
and the vector $H(t) \triangleq  (H_0(t), \dots, H_{k-1}(t))$. 
We observe that the number of queries in the system and the number of servers working at a higher rate at time $t$ can be written in terms of the system state $Z(t)$ as
\meq{2}{
&R(t) = \sum_{i=0}^{k-1}Y_i(t), &&{M_h(t)} = \sum_{i=0}^{k-1}H_{i}(t).
} 
It follows that in terms of the system state $Z(t)$, 
we can describe the tuple $(N(t), H(t), R(t))$ at all times $t$, 
and therefore we can compute the two performance metrics, namely the mean number of queries in the system $\bar{R}$ and the limiting power consumption $\bar{P}$.

We define the joint state at $i$th tandem queue as 
$Z_i(t) \triangleq  (Y_i(t), H_i(t))  \in \Z_+\times \set{0, \dots, n-i}$ and the aggregate vector for queues in tandem as 
\EQ{
Z(t) \triangleq (Z_0(t), \dots, Z_{k-1}(t)) \in \sZ \triangleq (\Z_+\times [n])^{\set{0, \dots, k-1}}.
} 
It follows that, it is sufficient to characterize the evolution of $Z(t)$ to characterize the number of queries in the system at any time~$t$. 

But, in order to evaluate the performance metrics of interest, it is necessary to study the stationary distribution of the system.
\begin{thm}
The aggregate process $(Z(t), t\in \R_+)$ is a continuous time Markov chain 
for any $(n,k)$ fork-join FCFS queueing system with Poisson arrivals and $n$ independent servers having independent exponential service and random rate selection policy. 
\end{thm}
\ieeeproof{
We first observe that aggregate system state takes countable and changes only at the external arrival or service completion at one of the servers. 
Since arrivals and service completions are all independent and have continuous memoryless distribution, 
it follows that
(a) we can have at most a single transition in an infinitesimal time interval, 
(b) the inter-transition times have continuous memoryless distribution, and 
(c) the jump probabilities would depend only on the previous state. 
Hence, the aggregate process $(Z(t), t \in \R_+)$ is a continuous time Markov chain~~\cite[Chapter 4]{Ross1995Wiley}.
}
\subsection{State transitions}
Given the current state $Z(t) = (y,h)$ at time $t$, 
the states transition only at the event of an external arrival or service completion at one of the stages of the tandem queue. 
We treat these two cases separately. 
\subsubsection{External arrival}
An external arrival changes the tandem queue occupancy of queries from $y \to a_0(y)$.  
If there are existing queries at stage~$0$, number of high rate servers remain unchanged. 
Otherwise, there are $N_0(a_0(y))$ idle servers that can potentially serve this query. 
Because of random rate selection policy, 
each of the idle servers can select a high rate server independently with probability $p$. 
Therefore, the number of high rate servers at stage~$0$ can change from $h_0 \to h_0 + B_0$ 
where $B_0$ is a Binomial random variable with parameters $(N_0(a_0(y)),p)$.  
The number of high rate servers at other stages remain unchanged. 

\subsubsection{Service completion}
A service completion of a query with $(i-1)$ response completions causes this query to depart from $(i-1)$th stage of the tandem queue for $i \le k$, 
and this transition results in the tandem queue occupancy to change from $y \to a_i(y)$. 
If the stage~$i$ has existing queries or if $i=k$,  
then the only change in number of high rate servers occur at the ones serving the stage~$(i-1)$. 
If the service completion occurs at a high rate server, 
then it can switch to a high or low rate server with probability $p$ and $(1-p)$ respectively. 
In the first case, the number of high rate servers remain unchanged, 
and in the second it is reduced by a unit amount at $(i-1)$th stage. 
Similarly, if the service completion occurs at a low rate server, 
then the number of high rate server increases by a unit  if the server switches to the higher rate, 
and it remains unchanged otherwise. 
The number of high rate servers at other stages remain unchanged. 

When $i < k$, a departure from $(i-1)$th stage of the tandem queue leads to an arrival at $i$th stage. 
If there were no queries at stage $i$, then newly arrived query can get serviced by a pool of $N_i(a_i(y))$ idle servers. 
Each of these idle servers can choose to serve at any of the two available rates in an \iid manner. 
Thus, when $y_i = 0$, then the number of high rate servers $h_i$ at stage $i$ transitions to the state $h_i + B_i$, 
where $B_i$ is a Binomial random variable with parameters $(N_i(a_i(y)),p)$.
 
\subsection{Transition Rates}
Next, we look at the transition rates corresponding to the state transitions listed in the previous subsection. 
Again, we treat the cases of external arrivals and service completions separately. 
We denote the state of the system at time $t$ by $Z(t) = (y,h)$, 
where sequence $y$ is the occupancy of the tandem queue, 
and sequence $h$ is the number of high rate servers active at each stage of the tandem queue. 
Recall that given the occupancy $y$ of the tandem queue, 
one can determine the number of pooled servers at stage~$i$ by $N_i(y)$. 
We will denote Bernoulli random variables with parameters $(N_i(a_i(y)), p)$ by $B_i$ at each stage~$i$, 
and we can write 
$
p_i(b) \triangleq P\set{B_i = b},\quad b \in \set{0, \dots, N_i(a_i(y))}. 
$

\subsubsection{External Arrival} 
Since the arrival process for queries into the system is Poisson with rate $\lambda$, 
we can write the transition rate from current state $z = (y,h)$ to next state $z' = (a_0(y), h')$ as  
\EQ{
Q(z,z') = 
\begin{cases}
\lambda, & h' = h, y_0 > 0,\\
\lambda p_0(b), & h' = h + be_0, y_0 = 0.
\end{cases}
}
\subsubsection{Service Completion} 
Given the current state $(y,h)$ at time $t$, 
there are $h_i$ high rate servers and $N_i(y)-h_i$ low rate servers serving the stage~$i$. 
Since all the service times are independent and memoryless, 
the random time for a service completion at stage~$i$ is also memoryless with rate $\tilde{\mu}_i \triangleq h_i\mu_1+(N_i(y)-h_i)\mu_0$. 
The corresponding rates for service completion due to a high rate server is $h_i\mu_1$, 
and due to a low rate server is $(N_i(y)-h_i)\mu_0$. 
Therefore, we can write the transition rate from current state $z = (y,h)$ to next state $z' = (a_{i+1}(y), h')$, 
when either $i=k-1$ or $y_{i+1} > 0$, as 
\EQ{
Q(z,z') = 
\begin{cases}
h_i\mu_1p + (N_i(y)-h_i)\mu_0(1-p), & h' = h,\\
h_i\mu_1(1-p), & h' = h-e_i,\\
(N_i(y)-h_i)\mu_0 p, & h' = h+e_i.
\end{cases}
}
When $i < k-1$ and the next stage of tandem queue is empty, i.e. $y_{i+1} = 0$, 
then we can write the transition rate $Q(z,z')$ from current state $z = (y,h)$ to next state $z' = (a_{i+1}(y), h')$, 
in terms of $\delta_h \triangleq h' - h$ 
as 
%\EQ{
%Q(z,(a_{i+1}(y), h'+be_{i+1})) = p_{i+1}(b)Q(z, (a_{i+1}(y), h'))
%}
\EQ{
\begin{cases}
p_{i+1}(b)(h_i\mu_1p + (N_i(y)-h_i)\mu_0(1-p)), & \delta_h = be_{i+1},\\
p_{i+1}(b)h_i\mu_1(1-p), & \delta_h = be_{i+1} -e_i,\\
p_{i+1}(b)(N_i(y)-h_i)\mu_0 p, & \delta_h = e_i+be_{i+1}.
\end{cases}
}

It follows that, we can identify the evolution of $(n,k)$ fork-join queue with probabilistic slowdown, 
as a continuous time Markov chain. 

Even though, we were able to completely characterize the evolution of this Markov process, 
finding the limiting distribution of this process remains intractable analytically. 
Therefore, to gain insights into this system, 
we resort to the following approximation. 
%%%%%%%%%%%%%%%%%%%%%%%%%%%%%%%%%%%%%%%%%
\section{Approximation}
\label{sec:Approx}
%%%%%%%%%%%%%%%%%%%%%%%%%%%%%%%%%%%%%%%%% 
The Markov process $Z(t)$ is non-reversible in general, 
and it is difficult to compute its equilibrium distribution. 
Therefore, we propose the following reversible Markov approximation for the occupancy vector $Y(t)$ and compute its equilibrioum distribution. 
\begin{appx}
Consider a reversible Markov process $\bY \triangleq (\bY(t) \in \R_+^{\set{0, \dots, k-1}}: t \in \R_+)$ of $k$ queues in tandem, 
where the external arrival process to queue~$0$ is Poisson with rate $\lambda$, 
and each level has an independent dedicated server with independent service time. 
The random service times at level~$i$ are \iid memoryless with rate $\bN_i\bmu$. 
If $\bpi_i$ is the marginal equilibrium distribution for occupancy at level~$i$, 
then 
\EQN{
\label{eqn:ApproxNumServers}
\bN_i = 
\begin{cases}
n-k+1, & i = k-1,\\
1 + \bN_{i+1}\bpi_{i+1}(0), & i < k-1.
\end{cases}
}
Further, the mean service rate $\bar{\mu}$ is the reciprocal of mean service time at each server, i.e.
\EQN{
\label{eqn:ApproxServiceRate}
\bar{\mu} = \left(\frac{p}{\mu_1} + \frac{1-p}{\mu_0}\right)^{-1}.
}
\end{appx}
\begin{rem} 
Recall that the evolution of occupancy vector $Y(t)$ depends on the number of high rate servers $H$ at each time $t$. 
We have two-fold approximations, 
we have replaced the effective service rate at level~$i$ by the reciprocal of mean service time at each server. 
Due to distributed probabilistic slowdown, each server can be in high or low rate with probability $p$ and $1-p$ respectively, 
in an \iid fashion at the beginning of each service. 
Hence, the mean service time is $1/\bmu$, where $\bmu$ is defined in Eq.~\eqref{eqn:ApproxServiceRate}. 
This reduces our problem to single-rate tandem queue, and we can use the independent queue approximation introduced in~\cite{Parag2017Infocom, Badita2019TIT}. 
Recall that a sequence of unpooled tandem queues with Poisson arrivals and memoryless service, 
have independent occupancies. 
In the independent queue approximation, 
the state-dependent number of parallel servers in the pooled tandem queue at any level~$i$ is replaced by its mean under the independent queue model. 
That is,  Eq.~\eqref{eqn:ApproxNumServers} is obtained by taking expectation of Eq.~\eqref{eqn:NumServersPerQueue} under the independent queue assumption. 
\end{rem}

\subsection{Equilibrium distribution of approximating Markov process}
It is easy to compute the equilibrium distribution of this approximating reversible Markov process $\bY$ and is given as follows. 
\begin{rem}\cite[Section~2.2, Page~33]{Kelly2011Cambridge}
\label{rem:ApproxMarkovProcess}
The Markov process $\bY$ is reversible, with equilibrium distribution denoted by $\bpi$ and 
\EQN{
\label{eqn:EqDistApprox}
\bpi(y) = \prod_{i=0}^{k-1}\bpi_i(y_i) = \prod_{i=0}^{k-1}(1-\rho_i)\rho_i^{y_i},
}
where the load on level~$i$ is denoted by $\rho_i \triangleq \frac{\lambda}{\bar{N}_i\bar{\mu}}$. 
The probability of $i$-th tandem queue being empty is given by $\bpi_i(0) = 1 - \frac{\lambda}{\bar{N}_i\bar{\mu}}$. 
\end{rem}
We observe that the equilibrium distribution $\bpi$ is computed in terms of the parameters $(\bN_0, \dots, \bN_{k-1})$, 
and these parameters are in turn computed in terms of equilibrium distribution $\bpi$. 
Now, we try to characterize an important quantity of interest in the computation of mean power consumption; the mean number of active servers.
This characterization is presented in the following Lemma for the approximate $(n,k)$ fork-join system governed by the reversible Markov chain $\bY(t)$. 
\begin{lem}
\label{lem:MeanActiveServers}
The mean number of active servers in the $(n,k)$ fork-join system is approximated as
\EQ{
\bar{M} = \sum_{i=0}^{k-1}\frac{\lambda (n-i)}{\bar{N}_i\bar{\mu}}\prod_{j=0}^{i-1} \left(1 - \frac{\lambda}{\bar{N}_j\bar{\mu}}\right), 
}
where the parameter
%$\bpi_i(0) = 1 - \frac{\lambda}{\bar{N}_i\bar{\mu}}$ with 
\EQ{
%\label{eqn:ApproxNumServersExplicit}
\bar{N}_i = (n - i ) - (k-1-i)\frac{\lambda}{\bar{\mu}},~i \in \set{0, \dots, k-1}.
}
\end{lem}
\begin{proof}
Note that the queries are served in an FCFS manner and there can be an idle server at time $t$ in the given system if and only if  $Y_j(t) = 0$ for all $j \le i$ for some $i \in \set{0, \dots, k-1}$.
Furthermore, if $(Y_j(t): j \le i)$ is zero it implies that all queries in the system has received atleast $i$ responses and hence at least $i$ servers are idle. 
%
%We observe that the $(n,k)$ fork-join system has idle servers at time $t$, if and only if $Y_j(t) = 0$ for all $j \le i$ for some $i \in \set{0, \dots, k-1}$. 
%In fact, if $(Y_j(t): j \le i)$ is zero, then it implies that there are no queries in the system that has less than $i+1$ response completions. 
%As the queries are served in a FCFS fashion this means that there are at least $i$ idle servers. 
On the other hand, non-zero $Y_i(t)$ indicates that there is at least one query that needs service from $n-i$ servers which implies that there are $n-i$ non-empty servers in the system. 
Thus, the number of idle servers $n-M(t)$ at any time $t$ is given by the first non-zero entry of $Y(t)$. 
That is, 
\EQ{
M(t) = n- \inf\set{i \le k-1: Y_i(t) > 0}.
}
We can define the equilibrium distribution of number of active servers as $p_i \triangleq  \lim_{t\to \infty}P\set{M(t) = n-i}$, 
where 
\EQ{
p_i = \lim_{t\to \infty}P\set{Y_i(t) > 0, \max_{j < i}Y_j(t)=0}.
}
Hence, we get 
\EQ{
\bar{M} = \sum_{i=0}^{k-1}(n-i)p_i.
}
Now, we provide an approximation for $p_i$ as follows.
Approximating the $(n,k)$ fork-join system evolution by the reversible Markov process $\bY(t)$, 
we get
\EQ{
p_i\approx  \lim_{t\to \infty}P\set{\bY_i(t) > 0, \bY_j(t) = 0, j < i}.
}
From the independence of tandem queues for the approximating Markov process $\bY(t)$ we get \EQ{
p_i = (1-\bpi_i(0))\prod_{j=0}^{i-1} \bpi_j(0),\text{ for }i \in \set{0, \dots, k-1}
}
where $\bpi_i(0)$ is provided in Remark~\ref{rem:ApproxMarkovProcess}.
Finally, we compute $\bN_i$ as follows.
From Eq.~\eqref{eqn:ApproxNumServers} we obtain, 
\EQ{
\bN_i = \E_{\bpi}N_i(y) = 
\begin{cases}
n-k+1, & i = k-1,\\
1 + \bN_{i+1}\bpi_{i+1}(0), & i < k-1.
\end{cases}
}
From Remark~\ref{rem:ApproxMarkovProcess}, we have $\pi_{i+1}(0) = 1 - (\lambda/\bN_{i+1}\bmu)$ for $i < k-1$.  
Substituting this in above equation, we get 
\EQ{
\bN_i = 
\begin{cases}
n-k+1, & i = k-1,\\
1 + \bN_{i+1} - \frac{\lambda}{\bmu}, & i < k-1.
\end{cases}
}
The result follows by recursive evaluation of the above equation from $i=k-1$ to $0$. 
\end{proof}
\begin{rem}
\label{rem:NumHighServers} 
We observe that $p_0 + p_1 + \dots + p_{k-1} = 1 -p_n$ where $p_n = \bpi_0(0) \dots \bpi_{k-1}(0)$. 
Using this relation, we can write the number of active servers as 
$
\bar{M} %&= n - \sum_{j=0}^{k-1}\bpi_0(0)\dots \bpi_j(0) - (n-k)p_n\\
= (n-k)(1-p_n) + \sum_{j=0}^{k-1}(1-\bpi_0(0)\dots \bpi_j(0)). 
$
Recall that the probability $\bpi_i(0) = 1 - \frac{\lambda}{\bN_i\bmu(p)}$.
We observe that the aggregate service rate $\bN_i\bmu(p)$ increases with the probability of high rate selection $p$
which leads to decrease in each $\bpi_i(0)$ for $i \in \set{0, \dots, k-1}$. 
Hence, the mean number of active servers $\bar{M}$ decreases. 
But, the mean number of active high rate servers $p\bar{M}$ may increase with $p$, if the increase in $p$ dominates the decrease in $\bar{M}$.
\end{rem}

\subsection{Computation of performance metrics}
Here, we will derive closed form expressions for mean number of queries and mean power consumption in the approximate $(n,k)$ fork-join system. 
From Little's Law~\cite[Chapter 3]{Ross1995Wiley}, the mean sojourn time in a queue is proportional to the mean number of queries at steady state. 
Therefore, we first compute the mean number of queries in the system. 
\begin{prop}[Mean number of queries] 
The mean number of queries in the approximate $(n,k)$ fork-join system is 
\EQ{
\bar{R} = \E_{\bpi}\sum_{i=0}^{k-1}\bY_i(t) = \sum_{i=0}^{k-1} \frac{\lambda}{(k-i) \left( \frac{n-i}{k-i} \bmu  - \lambda \right)}.
}
\end{prop}
\begin{proof} 
For the reversible process $\bY$, 
the number of queries at each level is independent with marginal distribution $\bpi_i$ at level~$i$. 
Further, each level~$i$ is an $M/M/1$ queue with Poisson arrival rate $\lambda$ and exponential service rate $\bN_i\bmu$, 
and hence the mean number of queries at level~$i$ is given by 
\EQ{
\E_{\bpi_i}\bY_i(t) = \frac{\lambda}{\bN_i\bmu - \lambda}. 
}
Substituting the value of parameter $\bN_i$ from Lemma~\ref{lem:MeanActiveServers} in the above equation, 
and summing the mean number of queries at all levels $i \in \set{0, \dots, k-1}$, 
we get the result. 
%\EQ{
%\E_{\bpi_i}\bY_i(t) = \frac{\lambda}{\bN_i\bmu - \lambda} = \frac{1}{(n-i)\frac{\bmu}{\lambda}- (k-i)}. 
%}
\end{proof}
Little's law gives us the mean sojourn time from the above Proposition as stated next.
\begin{cor}
The mean sojourn of the queries in the approximate $(n,k)$ fork-join system is 
\EQ{
\frac{\bar{R}}{\lambda} = \sum_{i=0}^{k-1} \frac{1}{ \left(  (n-i)\bmu- (k-i)\lambda \right)}.
} 
\end{cor}

\begin{rem}
We observe that the system is stable if $(k-i) \lambda < (n-i)\bmu$ for all $i \in \set{0, \dots, k-1}$.  
This is expected since when a query has response from $i$ servers, 
it still needs $(k-i)$ responses, and there are $(n-i)$ servers available to serve it. 
Recall that the mean service time at each server is $1/\bmu$. 

We further observe that the mean sojourn time of a query decreases with increase in mean service rate $\bmu$ whereas it increases with the arrival rate $\lambda$. 
Keeping $n/k$ fixed and taking large $k$ and $n$, we can approximate the mean sojourn time of a query by 
%\EQ{
%\frac{\bar{R}}{\lambda} \approx \frac{1}{\bmu-\lambda}\left[\ln\Big(1-\frac{k}{n}\Big) - \ln\Big(1 - \frac{k\lambda}{n\bmu}\Big)\right]. 
%}
\EQ{
\frac{\bar{R}}{\lambda} \approx \frac{1}{\bmu-\lambda}\left[\ln\Big(1 - \frac{k\lambda}{n\bmu}\Big) - \ln\Big(1-\frac{k}{n}\Big) \right]. 
}
From above equation it can be verified that for a fixed $\frac{\lambda}{\bmu}$, 
the mean sojourn time increases with $\frac{k}{n} \in [0, \min(1,\frac{\bmu}{\lambda}))$. 
%\red{Perhaps a line on behavior of $\bar{R}/\lambda$ as a function of system parameters for $k/n\approx 1$.}
\end{rem}
We now proceed to compute the mean power consumption in the approximate $(n,k)$ fork-join system. 
From Eq.~\eqref{eqn:MeanPower}, 
we know that mean power consumption can be written in terms of the mean number of active servers $\bar{M}$, 
and the mean number of high-rate servers $\bar{M}_h$. 
%\begin{lem}
%\label{lem:ConditionalMeanHighRateServers}
%Given the number of active servers $M(t)$ at any time $t$ in the  $(n,k)$-fork join system, 
%the conditional mean of number of high-rate servers in the system is given by 
%\EQ{
%\E[M_h(t)| M(t)] = M(t)p,
%}
%where $p$ is the selection probability of high rate. 
%\end{lem}
%\begin{proof}
%Without loss of generality, we choose to index the active servers as $1$ to $M(t)$ at time $t$. 
%Then, we can write the number of servers working at high-rate $\mu_1$ as $\sum_{j=1}^{M(t)}\xi_{j}$ where 
%%server $j$ is working on $i$th query, and 
%$\xi_{j}$ indicates that the active server $j$ is working at high-rate.  
%Since the mean of indicators $\xi_{ji}$ equals $p$,  the result follows from linearity of expectation. 
%\end{proof}
\begin{prop}
The mean power consumption in the approximate $(n,k)$ fork-join system with independent probabilistic slowdown is  
\EQ{
\bar{P} =\bar{M}(P_0+(P_1-P_0)p),
}
where $\bar{M}$ is the mean number of active servers in the system approximated as in Lemma~\ref{lem:MeanActiveServers}, and $p$ is the probability of selection of high rate. 
\end{prop}
Without loss of generality, we choose to index the active servers as $1$ to $M(t)$ at time $t$. 
Then, we can write the number of servers working at high-rate $\mu_1$ as $\sum_{j=1}^{M(t)}\xi_{j}$ where 
$\xi_{j}$ indicates that the active server $j$ is working at high-rate.  
Since the mean of indicators $\xi_{ji}$ equals $p$, from linearity of expectation we obtain that the conditional mean of number of high-rate servers in the $ (n,k)$-fork join system as 
\EQ{
\E[M_h(t)| M(t)] = M(t)p,
}
where $p$ is the selection probability of high rate.  
The result follows from taking expectation of conditional mean of high rate servers in the above equation, 
and substituting in Eq.~\eqref{eqn:MeanPower} for the mean power consumption. 
We can approximate the number of active servers $\bar{M}$, as in Lemma~\ref{lem:MeanActiveServers}.
\begin{rem} 
From the above proposition, 
we observe that the mean power consumption is linearly increasing in the mean number of active servers, 
and the probability of high rate selection $p$. 
However, the mean number of active servers depends on the parameter $p$, as it governs the mean service rate $\bmu$.  
For fixed arrival and service rates $\lambda, \mu_0, \mu_1$, 
when $p$ increases, 
the likelihood of servers working at the higher rate increases, 
and the mean number of active server reduces.  
\end{rem}

%%%%%%%%%%%%%%%%%%%%%%%%%%%%%%%%%%%%%%%%%
\section{Numerical studies}
\label{sec:Numerical} 
%%%%%%%%%%%%%%%%%%%%%%%%%%%%%%%%%%%%%%%%% 
The set of arrival rates $\lambda$ such that the query response time is stable, 
is called the \emph{stability region}. 
From the equilibrium distribution of the approximate system, 
we observe that the system is stable when each queue in tandem is stable, 
and we can write the stability region as $\cap_{i=0}^{K-1}\set{\frac{\lambda}{\bar{\mu}} < \bar{N}_i}$. 
Using Eq.~\eqref{eqn:ApproxServiceRate} for mean service rate $\bmu$ and parameter $\bN_i$ from Lemma~\ref{lem:MeanActiveServers}, 
we can equivalently write the stability region as 
\EQ{
\lambda < \min_i\frac{n-i}{k-i} \bar{\mu} = \frac{n}{k}\bmu.%\left(\frac{p}{\mu_1}+\frac{1-p}{\mu_0}\right)^{-1}.
}
We numerically evaluate the proposed approximation for random server slowdown for all ranges of arrival rates for which the tandem queueing system remains stable. 
To this end, we simulate $(n,k)$ fork-join system for total number of servers $n=20$ and number of sufficient responses $k = 18$. 
This corresponds to $k/n = 0.9$ or $90$th percentile of all responses. 
In a typical query response system, the query is sent to thousands of servers\footnote{https://www.computerweekly.com/news/2240088495/Single-Google-search-uses-1000-servers, 20 Feb 2009, viewed 21 May 2021.}.%Single Google search uses 1000 servers, 
{The implementation of our query response system is computationally demanding due to the overheads involved in keeping track of the query status at each server until the completion of $90$th percentile of responses. 
Due to computing limitations,
we provide the simulations only for a small number of servers. 
However, we would like to add that the proposed approximations get better with increasing number of servers and the plots show that the proposed approximation remains tight starting from number of servers as low as $20$. 
Therefore, the proposed approximation can be used to select the right operating point for large system case, 
when it is not feasible to simulate the system. 
Note that since we provide closed form expressions for the performance metrics under the approximation, 
the computational complexity involved in obtaining the tradeoff points via approximation is negligible. 
}

All $n$ servers are assumed to have independent and memoryless query response times, 
with high and low service rates being $\mu_1=\frac{k}{n}$ and $\mu_0 = 0.6\frac{k}{n}$ respectively. 
This choice of scaled service rates results in mean service rate 
\EQN{
\bmu(p) %= \frac{k}{n}(\frac{1}{p}+\frac{1}{0.6(1-p)})^{-1} 
= \frac{k}{n}\left(\frac{0.6p(1-p)}{p+0.6(1-p)}\right),
} 
for probability $p$ of selection of high rate at each server.  
This choice ensures that the maximum stability region is $\lambda < 1$ when all servers are chosen to run at high rate (when $p=1$). 
We compute the following empirical time averages: 
(a) number of queries,  (b) number of active servers, and (c) number of active high rate servers, in the system. 
We empirically verified that the empirical average of number of active high rate servers is $p$ times the empirical average of number of active servers. 
Hence, we only show the empirical average of number of active high rate servers in the following performance plots. 
%\red{Until Here.}
 
%For the above given values of system parameters, 
We first plot the mean number of active high rate servers $\bar{M}_h$ and the mean number of queries $\bar{R}$ as a function of arrival rate $\lambda \in [0, n\bmu(p)/k)$ in Fig.~\ref{fig:MeanNumHighRateServersVsArrivalRateWithFixedp} and Fig.~\ref{fig:MeanNumQueriesVsLambda} respectively, 
for two values of probability of high rate selection $p \in \set{0.3, 0.6}$. 
We observe that the performance curves obtained empirically remain very close to the one computed analytically under the proposed approximation, in both these cases. 
As more queries arrive, the number of active servers and hence the number of active high rate servers are expected to increase.
This is verified in Fig.~\ref{fig:MeanNumHighRateServersVsArrivalRateWithFixedp}.
For a fixed probability of high rate server selection $p$, 
both the system load and hence the number of queries in the system increase with the arrival rate. 
We observe this in Fig.~\ref{fig:MeanNumQueriesVsLambda}. 
In addition, we observe that the stability region is larger for larger value of $p$. 
This is expected since both the mean service rate $\bmu(p)$ and hence the stability region $\set{\lambda < n\bmu(p)/k}$, 
increase with $p$. 

%\blue{
%%is due to the activation of idle servers upon new query arrivals.
%We observe from Fig.~\ref{fig:MeanNumQueriesVsLambda} that choosing a larger value for high rate selection probability $p$ allows more number of queries in the system without causing instability.
%Further, we observe that the mean number of queries increase drastically from the point of arrival rates $\lambda > \frac{n}{k} \bar{\mu}$ which increases the mean sojourn time for each query.
%}
%%\red{[Fig.~\ref{fig:MeanNumQueriesVsLambda} explanation is missing here!]}

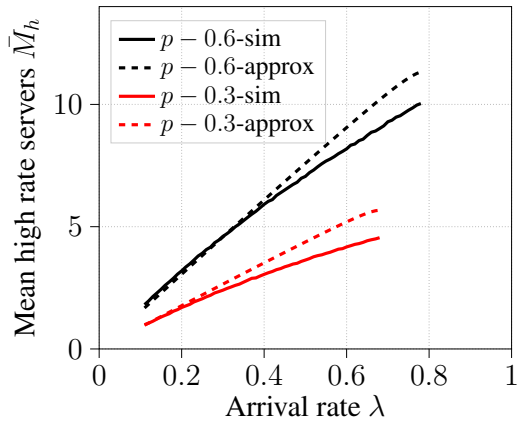
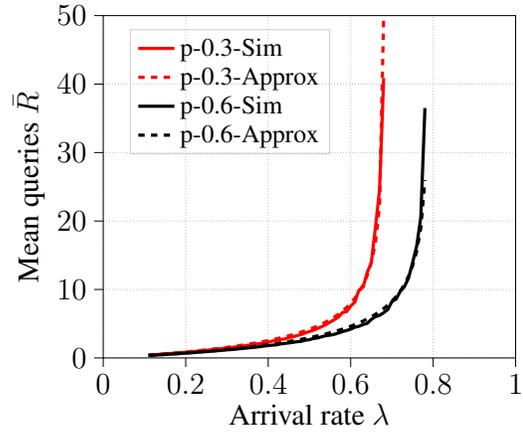
\begin{figure}[tbh]
\begin{subfigure}[t]{0.5\textwidth}
\centerline{\scalebox{0.8}{% This file was created by matplotlib2tikz v0.7.3.
\begin{tikzpicture}

\definecolor{color0}{rgb}{0.5,0,1}

\begin{axis}[
font=\Large,
legend cell align={left},
legend style={fill opacity=0.1, draw opacity=1, text opacity=1, draw=white!80.0!black, font=\large, at={(0.55,0.96)}},
tick align=outside,
tick pos=left,
x grid style={white!69.01960784313725!black, densely dotted},
x label style={at={(axis description cs:0.5,-0.02)},anchor=north},
xlabel={Arrival rate $\lambda$},
xmajorgrids,
xmin=0, xmax=1,
xtick style={color=black},
y grid style={white!69.01960784313725!black, densely dotted},
ylabel={Mean high rate servers $\bar{M}_h$},
ymajorgrids,
ymin=0, ymax=14,
ytick style={color=black}
]
\addplot [semithick, black, mark=, mark size=1.5, mark options={solid}, line width=1.5pt]
table {%
0.11 1.81217187828121
0.12 1.96257037429625
0.13 2.13943860561397
0.14 2.28624713752861
0.15 2.43636563634363
0.16 2.60324396756035
0.17 2.74158258417414
0.18 2.88986110138897
0.19 3.05844941550584
0.2 3.19188808111918
0.21 3.35088649113507
0.22 3.4819451805482
0.23 3.63911360886391
0.24 3.80150198498014
0.25 3.9130408695913
0.26 4.06514934850655
0.27 4.16179838201612
0.28 4.34360656393434
0.29 4.4505354946451
0.3 4.60232397676017
0.31 4.71548284517156
0.32 4.84320156798434
0.33 4.99533004669954
0.34 5.11020889791103
0.35 5.24618753812466
0.36 5.38051619483808
0.37 5.5100748992511
0.38 5.62782372176279
0.39 5.77126228737709
0.4 5.91434085659143
0.41 6.0352896471035
0.42 6.10662893371069
0.43 6.27793722062781
0.44 6.38953610463893
0.45 6.50074499255002
0.46 6.61435385646144
0.47 6.72679273207267
0.48 6.81593184068161
0.49 6.97954020459793
0.5 7.05563944360556
0.51 7.21059789402104
0.52 7.31671683283168
0.53 7.45223547764523
0.54 7.5750142498575
0.55 7.68333316666833
0.56 7.76553234467655
0.57 7.86388136118638
0.58 7.98306016939829
0.59 8.07756922430764
0.6 8.18215817841818
0.61 8.33364666353338
0.62 8.38115618843809
0.63 8.51929480705194
0.64 8.57885421145779
0.65 8.73194268057294
0.66 8.8618613813863
0.67 8.91137088629128
0.68 9.01470985290143
0.69 9.12253877461238
0.7 9.290467095329
0.71 9.35384646153553
0.72 9.47933520664805
0.73 9.53722462775367
0.74 9.65014349856491
0.75 9.76104238957608
0.76 9.85529144708563
0.77 9.97068029319698
0.78 10.0359296407035
};
\addlegendentry{$p-0.6$-sim}
\addplot [semithick, black, dashed, mark=, mark size=1.5, mark options={solid}, line width=1.5pt]
table {%
0.11 1.67683846788879
0.12 1.82988645207345
0.13 1.98297843781058
0.14 2.13610193127949
0.15 2.28924467628385
0.16 2.4423946629431
0.17 2.59554013634536
0.18 2.7486696051225
0.19 2.90177184990203
0.2 3.05483593158376
0.21 3.20785119938171
0.22 3.36080729856306
0.23 3.51369417780625
0.24 3.66650209608909
0.25 3.81922162900507
0.26 3.97184367439125
0.27 4.12435945713481
0.28 4.2767605330058
0.29 4.42903879134158
0.3 4.58118645638343
0.31 4.73319608703571
0.32 4.88506057478468
0.33 5.03677313947393
0.34 5.18832732258775
0.35 5.3397169776397
0.36 5.49093625720079
0.37 5.64197959602738
0.38 5.79284168966153
0.39 5.94351746777259
0.4 6.09400206138573
0.41 6.24429076299538
0.42 6.39437897838455
0.43 6.54426216875772
0.44 6.69393578153565
0.45 6.84339516784531
0.46 6.99263548435083
0.47 7.14165157659451
0.48 7.29043784042474
0.49 7.4389880573482
0.5 7.58729519871301
0.51 7.73535119245157
0.52 7.88314664460608
0.53 8.03067050592384
0.54 8.17790967129472
0.55 8.32484849651178
0.56 8.47146821248313
0.57 8.61774621121277
0.58 8.76365517002698
0.59 8.90916196982449
0.6 9.05422634835643
0.61 9.19879920888432
0.62 9.3428204752813
0.63 9.48621634252089
0.64 9.62889570994848
0.65 9.77074549324036
0.66 9.9116243724013
0.67 10.0513543189982
0.68 10.1897089072963
0.69 10.3263968653051
0.7 10.4610384074243
0.71 10.5931303182267
0.72 10.7219929557999
0.73 10.8466871448799
0.74 10.9658788173279
0.75 11.0776084636769
0.76 11.1788767836397
0.77 11.2648492025971
0.78 11.3271958941609
};
\addlegendentry{$p-0.6$-approx}

\addplot [semithick, red, mark=, mark size=1.5, mark options={solid}, line width=1.5pt]
table {%
0.11 0.991670083299167
0.12 1.05696943030569
0.13 1.13143868561313
0.14 1.22384776152237
0.15 1.30468695313044
0.16 1.37657623423764
0.17 1.44226557734421
0.18 1.54319456805433
0.19 1.60651393486064
0.2 1.69590304096959
0.21 1.76058239417605
0.22 1.82478175218247
0.23 1.92419075809241
0.24 1.98757012429875
0.25 2.05249947500523
0.26 2.13184868151319
0.27 2.19498805011949
0.28 2.28975710242897
0.29 2.34471655283445
0.3 2.40287597124029
0.31 2.47473525264745
0.32 2.53462465375342
0.33 2.59977400225997
0.34 2.64959350406497
0.35 2.73033269667305
0.36 2.80800191998082
0.37 2.88839111608882
0.38 2.90444095559043
0.39 2.98170018299817
0.4 3.05033949660503
0.41 3.11482885171148
0.42 3.16830831691683
0.43 3.24216757832421
0.44 3.28655713442863
0.45 3.35875641243588
0.46 3.40849591504083
0.47 3.47176528234715
0.48 3.5030149698503
0.49 3.57971420285797
0.5 3.63555364446353
0.51 3.69487305126948
0.52 3.76967230327696
0.53 3.81640183598164
0.54 3.85028149718503
0.55 3.90844091559084
0.56 3.95381046189538
0.57 4.0164098359017
0.58 4.0797792022079
0.59 4.10193898061021
0.6 4.17238827611725
0.61 4.2223877761222
0.62 4.27685723142766
0.63 4.30783692163079
0.64 4.35831641683581
0.65 4.42452575474241
0.66 4.46573534264655
0.67 4.49995500045
0.68 4.54964450355498
};
\addlegendentry{$p-0.3$-sim}

\addplot [semithick, red, dashed, mark=, mark size=1.5, mark options={solid}, line width=1.5pt]
table {%
0.11 0.971343756246333
0.12 1.0599912757914
0.13 1.14865256264308
0.14 1.23731828982326
0.15 1.32597935991454
0.16 1.41462691396656
0.17 1.50325234024089
0.18 1.59184728271848
0.19 1.68040364928153
0.2 1.76891361946629
0.21 1.85736965166622
0.22 1.94576448964455
0.23 2.03409116819178
0.24 2.12234301773595
0.25 2.21051366768108
0.26 2.29859704821103
0.27 2.38658739025061
0.28 2.47447922322291
0.29 2.56226737017748
0.3 2.64994693978899
0.31 2.73751331463439
0.32 2.82496213504781
0.33 2.91228927771965
0.34 2.99949082804595
0.35 3.08656304503734
0.36 3.17350231735557
0.37 3.26030510874662
0.38 3.34696789076782
0.39 3.43348706023972
0.4 3.51985883826494
0.41 3.60607914690621
0.42 3.69214345865325
0.43 3.77804661256181
0.44 3.8637825893181
0.45 3.94934423532976
0.46 4.03472292307262
0.47 4.11990813104638
0.48 4.20488692139978
0.49 4.28964328596501
0.5 4.37415732117921
0.51 4.45840417777616
0.52 4.54235271005018
0.53 4.62596371853226
0.54 4.70918763361538
0.55 4.7919614170474
0.56 4.87420434816891
0.57 4.95581218620065
0.58 5.03664891231873
0.59 5.11653477027664
0.6 5.1952284790151
0.61 5.2723999610739
0.62 5.34758704095211
0.63 5.42012382825608
0.64 5.48901640735755
0.65 5.55271410043662
0.66 5.60865706638789
0.67 5.65229508964599
0.68 5.67468128153597
};
\addlegendentry{$p-0.3$-approx}
\end{axis}

\end{tikzpicture}}}
\subcaption{Mean number of high rate servers $\bar{M}_h$.}
\label{fig:MeanNumHighRateServersVsArrivalRateWithFixedp}
\end{subfigure}%
\begin{subfigure}[t]{0.5\textwidth}
\centerline{\scalebox{0.8}{% This file was created by matplotlib2tikz v0.7.3.
\begin{tikzpicture}

\definecolor{color0}{rgb}{0.12156862745098,0.466666666666667,0.705882352941177}
\definecolor{color1}{rgb}{1,0.498039215686275,0.0549019607843137}
\definecolor{color2}{rgb}{0.172549019607843,0.627450980392157,0.172549019607843}
\definecolor{color3}{rgb}{0.83921568627451,0.152941176470588,0.156862745098039}
\definecolor{color4}{rgb}{0.580392156862745,0.403921568627451,0.741176470588235}
\definecolor{color5}{rgb}{0.549019607843137,0.337254901960784,0.294117647058824}
\definecolor{color6}{rgb}{0.890196078431372,0.466666666666667,0.76078431372549}
\definecolor{color7}{rgb}{0.737254901960784,0.741176470588235,0.133333333333333}
\definecolor{color8}{rgb}{0.0901960784313725,0.745098039215686,0.811764705882353}

\begin{axis}[
font=\Large,
legend cell align={left},
legend style={fill opacity=0.1, draw opacity=1, text opacity=1, draw=white!80.0!black, font=\large, at={(0.55,0.96)}},
tick align=outside,
tick pos=left,
x grid style={white!69.01960784313725!black, densely dotted},
x label style={at={(axis description cs:0.5,-0.02)},anchor=north},
xlabel={Arrival rate $\lambda$},
xmajorgrids,
xmin=0, xmax=1,
xtick style={color=black},
y grid style={white!69.01960784313725!black, densely dotted},
ylabel={Mean queries $\bar{R}$},
ymajorgrids,
ymin=0, ymax=50,
ytick style={color=black}
]
\addplot [semithick, red, mark=, mark size=1.5, mark options={solid}, line width=1.5pt]
table {%
0.11 0.424145758542413
0.12 0.462945370546295
0.13 0.500624993750065
0.14 0.554274457255423
0.15 0.594174058259421
0.16 0.641783582164178
0.17 0.679083209167904
0.18 0.745422545774541
0.19 0.784252157478416
0.2 0.841111588884108
0.21 0.890571094289058
0.22 0.94372056279437
0.23 1.01499985000147
0.24 1.06604933950661
0.25 1.12004879951201
0.26 1.17531824681752
0.27 1.25144748552515
0.28 1.32045679543203
0.29 1.37789622103777
0.3 1.44646553534464
0.31 1.51576484235155
0.32 1.60123398766012
0.33 1.66318336816631
0.34 1.73629263707362
0.35 1.8420415795842
0.36 1.92985070149298
0.37 2.04458955410443
0.38 2.08654913450864
0.39 2.2240477595224
0.4 2.31807681923179
0.41 2.45953540464594
0.42 2.58054419455808
0.43 2.70206297937019
0.44 2.76983230167698
0.45 2.96637033629664
0.46 3.08341916580831
0.47 3.27066729332707
0.48 3.41802581974178
0.49 3.65762342376576
0.5 3.86586134138658
0.51 4.05725942740569
0.52 4.32862671373295
0.53 4.56969430305695
0.54 4.87134128658721
0.55 5.23427765722349
0.56 5.57938420615798
0.57 5.95849041509584
0.58 6.62832371676285
0.59 6.76015239847597
0.6 7.57958420415797
0.61 8.15109848901514
0.62 9.89769102308975
0.63 10.3418865811341
0.64 12.4674253257468
0.65 13.9184408155918
0.66 19.0619493805062
0.67 24.569064309357
0.68 40.8825711742886
};
\addlegendentry{p-0.3-Sim}

\addplot [semithick, red, dashed, mark=, mark size=1.5, mark options={solid}, line width=1.5pt]
table {%
0.11 0.427985686434393
0.12 0.47290901794032
0.13 0.519021423263566
0.14 0.566375206395754
0.15 0.615025964010093
0.16 0.665032858324826
0.17 0.716458918239348
0.18 0.769371372282823
0.19 0.823842017437909
0.2 0.879947628514465
0.21 0.93777041346727
0.22 0.997398520899328
0.23 1.05892660699431
0.24 1.1224564703105
0.25 1.18809776428379
0.26 1.25596879897894
0.27 1.32619744565731
0.28 1.39892216017369
0.29 1.47429314417232
0.3 1.55247366664484
0.31 1.63364157280068
0.32 1.71799101258156
0.33 1.80573442778929
0.34 1.89710484502483
0.35 1.99235853189806
0.36 2.09177808684105
0.37 2.19567604911055
0.38 2.30439913622427
0.39 2.41833324252348
0.4 2.53790936667037
0.41 2.66361068025521
0.42 2.79598100788228
0.43 2.93563506614314
0.44 3.08327091188363
0.45 3.23968518934227
0.46 3.40579195594504
0.47 3.58264612972603
0.48 3.77147297040255
0.49 3.97370553123292
0.5 4.19103277784424
0.51 4.42546218671978
0.52 4.67940231004687
0.53 4.95577335729193
0.54 5.25815786373107
0.55 5.59100998823302
0.56 5.95995271843412
0.57 6.37221068541686
0.58 6.83725916684734
0.59 7.36783125020957
0.6 7.98154607262969
0.61 8.70367510271726
0.62 9.57214064812061
0.63 10.6472882838046
0.64 12.0331033133224
0.65 13.9305937752652
0.66 16.8063886266318
0.67 22.2014011431288
0.68 50.0835170483881
};
\addlegendentry{p-0.3-Approx}

\addplot [semithick, black, mark=, mark size=1.5, mark options={solid}, line width=1.5pt]
table {%
0.11 0.353576464235359
0.12 0.385306146938529
0.13 0.425795742042579
0.14 0.463465365346347
0.15 0.501484985150143
0.16 0.541404585954137
0.17 0.574304256957428
0.18 0.6180138198618
0.19 0.662953370466306
0.2 0.699143008569921
0.21 0.740902590974092
0.22 0.784152158478419
0.23 0.835941640583593
0.24 0.886051139488605
0.25 0.925960740392597
0.26 0.974130258697412
0.27 1.00734992650071
0.28 1.07600923990758
0.29 1.12405875941241
0.3 1.17512824871751
0.31 1.22105778942211
0.32 1.27706722932771
0.33 1.33675663243368
0.34 1.39721602783972
0.35 1.4639253607464
0.36 1.51091489085109
0.37 1.57514424855751
0.38 1.64782352176478
0.39 1.73135268647314
0.4 1.81014189858101
0.41 1.86833131668683
0.42 1.91533084669153
0.43 2.04573954260459
0.44 2.11753882461177
0.45 2.19210807891919
0.46 2.2761372386276
0.47 2.36557634423654
0.48 2.42495575044249
0.49 2.59099409005908
0.5 2.67351326486733
0.51 2.84303156968429
0.52 2.92124078759212
0.53 3.1032489675103
0.54 3.24934750652495
0.55 3.33072669273307
0.56 3.44530554694451
0.57 3.63777362226376
0.58 3.8439915600844
0.59 4.009279907201
0.6 4.20609793902057
0.61 4.47858521414785
0.62 4.61990380096188
0.63 4.87981120188795
0.64 5.05307946920525
0.65 5.62581374186259
0.66 6.02599974000261
0.67 6.22955770442297
0.68 6.59720402795973
0.69 7.08246917530824
0.7 8.22826771732262
0.71 8.59806401935982
0.72 9.70790292097065
0.73 10.390036099639
0.74 11.6291337086629
0.75 13.8974110258897
0.76 16.3217367826323
0.77 20.539984600154
0.78 36.4928650713491
};
\addlegendentry{p-0.6-Sim}

\addplot [semithick, black, dashed, mark=, mark size=1.5, mark options={solid}, line width=1.5pt]
table {%
0.11 0.362723225308287
0.12 0.399994558199717
0.13 0.438093724384232
0.14 0.477051156117066
0.15 0.516898881318052
0.16 0.557670633351341
0.17 0.599401970217354
0.18 0.642130404128778
0.19 0.68589554256032
0.2 0.730739241996385
0.21 0.776705775754468
0.22 0.823842017437909
0.23 0.87219764177347
0.24 0.921825344821332
0.25 0.972781085812674
0.26 1.02512435317929
0.27 1.07891845769808
0.28 1.13423085608945
0.29 1.19113350889383
0.3 1.24970327701692
0.31 1.31002236199827
0.32 1.372178795838
0.33 1.43626698713658
0.34 1.50238833139081
0.35 1.57065189458133
0.36 1.64117518072611
0.37 1.71408499591526
0.38 1.78951842355266
0.39 1.86762392819404
0.4 1.948562608597
0.41 2.03250962452164
0.42 2.11965582661419
0.43 2.21020962459313
0.44 2.30439913622427
0.45 2.40247466859133
0.46 2.50471159442959
0.47 2.61141370043444
0.48 2.72291710233824
0.49 2.8395948443097
0.5 2.96186232941204
0.51 3.09018376556128
0.52 3.22507986054351
0.53 3.36713706419796
0.54 3.51701874149518
0.55 3.6754787749645
0.56 3.84337825030933
0.57 4.02170609196522
0.58 4.21160481079436
0.59 4.41440294168266
0.6 4.63165634209919
0.61 4.8652013835417
0.62 5.11722433932345
0.63 5.39035318747864
0.64 5.68778099607592
0.65 6.01343470872127
0.66 6.37221068541685
0.67 6.77031095247999
0.68 7.21573592666553
0.69 7.71902870145731
0.7 8.29444030811657
0.71 8.96183392830772
0.72 9.7499636902982
0.73 10.7025016332236
0.74 11.890090876025
0.75 13.4373490311459
0.76 15.5940670841044
0.77 18.9785088207294
0.78 25.9680247723392
};
\addlegendentry{p-0.6-Approx}

\end{axis}

\end{tikzpicture}}}
\subcaption{Mean number of queries $\bar{R}$.}
\label{fig:MeanNumQueriesVsLambda}
\end{subfigure}%
\caption{Performance as a function of arrival rate $\lambda$ for probability of high rate selection $p \in \set{0.3, 0.6}$.}
\end{figure}
We next plot the mean number of active high rate servers $\bar{M}_h$ and the mean number of queries $\bar{R}$ as a function of high rate selection probability $p$ in Fig.~\ref{fig:avgHighRateServersVspWithFixedLambda} and Fig.~\ref{fig:avgNumCustomersVspWithFixedLambda} respectively, 
for different values of arrival rate $\lambda \in \set{0.1, 0.2, 0.3, 0.4, 0.5}$.  
For each arrival rate $\lambda$, the range of probability $p$ is limited by the stability region $\lambda \le \frac{n}{k}\bmu(p)$.
As the probability of high rate selection $p$ increases, 
we expect the number of high rate servers to increase for a fixed arrival rate $\lambda$, and this is shown in Fig.~\ref{fig:avgHighRateServersVspWithFixedLambda}.
Further, this cause the queries to get serviced faster, 
leading to reduced mean number of queries in the system as captured in Fig.~\ref{fig:avgNumCustomersVspWithFixedLambda}. 
That is, for a given query arrival rate $\lambda$, the power consumption increases whereas the mean sojourn time decreases with increase in high rate selection probability $p$, 
and allows us to choose a tradeoff point.
\begin{figure}[tbh]
\begin{subfigure}[t]{0.5\textwidth}
\centerline{\scalebox{0.8}{% This file was created by tikzplotlib v0.9.1.
\begin{tikzpicture}

\definecolor{color0}{rgb}{0.12156862745098,0.466666666666667,0.705882352941177}
\definecolor{color1}{rgb}{1,0.498039215686275,0.0549019607843137}
\definecolor{color2}{rgb}{0.172549019607843,0.627450980392157,0.172549019607843}
\definecolor{color3}{rgb}{0.83921568627451,0.152941176470588,0.156862745098039}
\definecolor{color4}{rgb}{0.580392156862745,0.403921568627451,0.741176470588235}
\definecolor{color5}{rgb}{0.549019607843137,0.337254901960784,0.294117647058824}
\definecolor{color6}{rgb}{0.890196078431372,0.466666666666667,0.76078431372549}
\definecolor{color7}{rgb}{0.737254901960784,0.741176470588235,0.133333333333333}
\definecolor{color8}{rgb}{0.0901960784313725,0.745098039215686,0.811764705882353}

\begin{axis}[
font=\Large,
legend cell align={left},
legend style={fill opacity=0.1, draw opacity=1, text opacity=1, at={(0.03,0.97)}, anchor=north west, draw=white!80!black},
tick align=outside,
tick pos=left,
x grid style={white!69.01960784313725!black, densely dotted},
xlabel={high rate selection probability $p$},
x label style={at={(axis description cs:0.5,-0.02)},anchor=north},
xmajorgrids,
xmin=-0.05, xmax=1.05,
xtick style={color=black},
xtick={-0.2,0,0.2,0.4,0.6,0.8,1,1.2},
xticklabels={−0.2,0.0,0.2,0.4,0.6,0.8,1.0,1.2},
%y grid style={white!69.01960784313725!black, densely dotted},
ylabel={Mean high rate servers $\bar{M}_h$},
ymajorgrids,
ymin=-0.517186828131705, ymax=12,
ytick style={color=black}
]
\addplot [semithick, color0, mark=, mark size=1.5, mark options={solid}, line width=1.5pt]
table {%
0 0
0.05 0.156298437015628
0.1 0.307726922730774
0.15 0.461635383646163
0.2 0.601763982360183
0.25 0.746242537574634
0.3 0.887831121688786
0.35 1.02331976680229
0.4 1.17553824461756
0.45 1.29340706592935
0.5 1.4190458095419
0.55 1.52348476515234
0.6 1.6671133288667
0.65 1.79416205837941
0.7 1.92741072589274
0.75 2.0151398486015
0.8 2.13497865021354
0.85 2.18248817511823
0.9 2.280507194928
0.95 2.40988590114098
1 2.50352496475037
};
\addlegendentry{$\lambda$-0.1}

\addplot [semithick, color1, mark=, mark size=1.5, mark options={solid}, line width=1.5pt]
table {%
0 0
0.05 0.289957100428995
0.1 0.587514124858747
0.15 0.869261307386924
0.2 1.14199858001421
0.25 1.41814581854182
0.3 1.69221307786921
0.35 1.95087049129508
0.4 2.20644793552063
0.45 2.49282507174928
0.5 2.72126278737209
0.55 2.95051049489504
0.6 3.19010809891903
0.65 3.41471585284147
0.7 3.61855381446185
0.75 3.8660713392866
0.8 4.07258927410717
0.85 4.23778762212383
0.9 4.51922480775187
0.95 4.6427635723642
1 4.77828221717777
};
\addlegendentry{$\lambda$-0.2};
\addplot [semithick, color2, mark=, mark size=1.5, mark options={solid}, line width=1.5pt]
table {%
0 0
0.05 0.414585854141458
0.1 0.821441785582143
0.15 1.23397766022339
0.2 1.62230377696222
0.25 2.00291997080029
0.3 2.40157598424014
0.35 2.79910200897988
0.4 3.17333826661735
0.45 3.55551444485555
0.5 3.86712132878671
0.55 4.2780072199278
0.6 4.60012399875998
0.65 4.9409505904941
0.7 5.29321706782931
0.75 5.5859841401586
0.8 5.88321116788833
0.85 6.2192578074219
0.9 6.46502534974648
0.95 6.7959520404796
1 6.99485005149947
};
\addlegendentry{$\lambda$-0.3}
\addplot [semithick, color3, mark=, mark size=1.5, mark options={solid}, line width=1.5pt]
table {%
0 0
0.05 0.519144808551903
0.1 1.03273967260327
0.15 1.53666463335366
0.2 2.04946950530496
0.25 2.54570454295456
0.3 3.08748912510875
0.35 3.52315476845232
0.4 3.9979400205998
0.45 4.47602523974749
0.5 4.99348006519934
0.55 5.43416565834349
0.6 5.87017129828703
0.65 6.35068649313507
0.7 6.7529324706753
0.75 7.22242777572224
0.8 7.59921400785991
0.85 8.00172998270013
0.9 8.41664583354168
0.95 8.78236217637809
1 9.05407945920544
};
\addlegendentry{$\lambda$-0.4}
\addplot [semithick, color4, mark=, mark size=1.5, mark options={solid}, line width=1.5pt]
table {%
0 0
0.05 0.605763942360587
0.1 1.21992780072201
0.15 1.82778172218277
0.2 2.43124568754308
0.25 3.0509694903051
0.3 3.61721382786172
0.35 4.22791772082286
0.4 4.79081209187915
0.45 5.39923600763999
0.5 5.98094019059812
0.55 6.56176438235617
0.6 7.08670913290864
0.65 7.63142368576314
0.7 8.16253837461625
0.75 8.67302326976725
0.8 9.2567274327257
0.85 9.69951300487001
0.9 10.2050779492206
0.95 10.6707832921671
1 11.1075489245108
};
\addlegendentry{$\lambda$-0.5}
\addplot [semithick, color0, dashed, mark=, mark size=1.5, mark options={solid}, line width=1.5pt]
table {%
0 0
0.05 0.163905149446792
0.1 0.321095300596137
0.15 0.471571077848174
0.2 0.615333181505945
0.25 0.752382388239256
0.3 0.882719551543047
0.35 1.00634560219026
0.4 1.12326154867928
0.45 1.23346847767584
0.5 1.3369675544495
0.55 1.43376002330471
0.6 1.52384720800635
0.65 1.60723051219992
0.7 1.68391141982631
0.75 1.75389149553114
0.8 1.81717238506874
0.85 1.8737558157008
0.9 1.9236435965896
0.95 1.96683761918603
1 2.00333985761221
};
%\addlegendentry{0.1-Approx}
\addplot [semithick, color1, dashed, mark=, mark size=1.5, mark options={solid}, line width=1.5pt]
table {%
0 0
0.05 0.328310104903977
0.1 0.643232523485115
0.15 0.944760341032024
0.2 1.23288707727836
0.25 1.50760669682536
0.3 1.76891361946629
0.35 2.01680273040249
0.4 2.25126939034139
0.45 2.47230944546801
0.5 2.6799192372825
0.55 2.87409561229666
0.6 3.05483593158376
0.65 3.22213808017609
0.7 3.37600047630566
0.75 3.51642208048413
0.8 3.64340240441823
0.85 3.75694151975786
0.9 3.85704006667428
0.95 3.94369926226636
1 4.01692090879304
};
%\addlegendentry{0.2-Approx}
\addplot [semithick, color2, dashed, mark=, mark size=1.5, mark options={solid}, line width=1.5pt]
table {%
0 0
0.05 0.49132967252517
0.1 0.962828359491984
0.15 1.4144671103228
0.2 1.84621780581369
0.25 2.258053195584
0.3 2.64994693978899
0.35 3.02187365431831
0.4 3.37380895880259
0.45 3.70572952683637
0.5 4.01761313790134
0.55 4.30943873053947
0.6 4.58118645638343
0.65 4.83283773470223
0.7 5.06437530716473
0.75 5.27578329256235
0.8 5.46704724126728
0.85 5.63815418923235
0.9 5.78909271136604
0.95 5.91985297413917
1 6.03042678730144
};
%\addlegendentry{0.3-Approx}
\addplot [semithick, color3, dashed, mark=, mark size=1.5, mark options={solid}, line width=1.5pt]
table {%
0 0
0.05 0.651745783585291
0.1 1.2775418789457
0.15 1.87731739393024
0.2 2.45100481846442
0.25 2.99853942601123
0.3 3.51985883826494
0.35 4.01490272089565
0.4 4.48361258468616
0.45 4.92593167146515
0.5 5.34180490820505
0.55 5.73117891578218
0.6 6.09400206138573
0.65 6.43022454555654
0.7 6.73979851644388
0.75 7.02267820517374
0.8 7.27882007728575
0.85 7.5081829960688
0.9 7.7107283943445
0.95 7.88642045184207
1 8.03522627580268
};
%\addlegendentry{0.4-Approx}
\addplot [semithick, color4, dashed, mark=, mark size=1.5, mark options={solid}, line width=1.5pt]
table {%
0 0
0.05 0.808016585993264
0.1 1.58479915207258
0.15 2.33002994385452
0.2 3.04346070719152
0.25 3.72489171102647
0.3 4.37415732117921
0.35 4.99111594235457
0.4 5.57564294623864
0.45 6.12762568515651
0.5 6.64695998924998
0.55 7.13354773521299
0.6 7.58729519871301
0.65 8.00811198551074
0.7 8.39591039283147
0.75 8.75060509187115
0.8 9.07211305017215
0.85 9.36035363265728
0.9 9.6152488347691
0.95 9.83672361202544
1 10.0247062784553
};
%\addlegendentry{0.5-Approx}
\end{axis}

\end{tikzpicture}}}
\subcaption{Mean number of high rate servers $\bar{M}_h$.}
\label{fig:avgHighRateServersVspWithFixedLambda}
\end{subfigure}%
\begin{subfigure}[t]{0.5\textwidth}
\centerline{\scalebox{0.8}{% This file was created by tikzplotlib v0.9.1.
\begin{tikzpicture}

\definecolor{color0}{rgb}{0.12156862745098,0.466666666666667,0.705882352941177}
\definecolor{color1}{rgb}{1,0.498039215686275,0.0549019607843137}
\definecolor{color2}{rgb}{0.172549019607843,0.627450980392157,0.172549019607843}
\definecolor{color3}{rgb}{0.83921568627451,0.152941176470588,0.156862745098039}
\definecolor{color4}{rgb}{0.580392156862745,0.403921568627451,0.741176470588235}
\definecolor{color5}{rgb}{0.549019607843137,0.337254901960784,0.294117647058824}
\definecolor{color6}{rgb}{0.890196078431372,0.466666666666667,0.76078431372549}
\definecolor{color7}{rgb}{0.737254901960784,0.741176470588235,0.133333333333333}
\definecolor{color8}{rgb}{0.0901960784313725,0.745098039215686,0.811764705882353}

\begin{axis}[
font=\Large,
legend cell align={left},
legend style={fill opacity=0.1, draw opacity=1, text opacity=1, draw=white!80!black},
tick align=outside,
tick pos=left,
x grid style={white!69.01960784313725!black, densely dotted},
xlabel={high rate selection probability $p$},
x label style={at={(axis description cs:0.5,-0.02)},anchor=north},
xmajorgrids,
xmin=-0.05, xmax=1.05,
xtick style={color=black},
xtick={-0.2,0,0.2,0.4,0.6,0.8,1,1.2},
xticklabels={−0.2,0.0,0.2,0.4,0.6,0.8,1.0,1.2},
y grid style={white!69.01960784313725!black, densely dotted},
ylabel={Mean number of queries $\bar{R}$},
ymajorgrids,
ymin=-0.0503470764311955, ymax=6.5,
ytick style={color=black}
]
\addplot [semithick, color0, mark=, mark size=1.5, mark options={solid}, line width=1.5pt]
table {%
0 0.430285697143028
0.05 0.423995760042398
0.1 0.411005889941101
0.15 0.406975930240697
0.2 0.394066059339408
0.25 0.384376156238439
0.3 0.376816231837683
0.35 0.363626363736366
0.4 0.363156368436312
0.45 0.349816501834979
0.5 0.340946590534096
0.55 0.328256717432826
0.6 0.322026779732204
0.65 0.312806871931278
0.7 0.308616913830861
0.75 0.294587054129458
0.8 0.285307146928526
0.85 0.270217297827021
0.9 0.261497385026145
0.95 0.256447435525644
1 0.248597514024861
};
\addlegendentry{$\lambda$-0.1}
\addplot [semithick, color1, mark=, mark size=1.5, mark options={solid}, line width=1.5pt]
table {%
0 0.976960230397696
0.05 0.949030509694902
0.1 0.929640703592964
0.15 0.909790902090979
0.2 0.886171138288617
0.25 0.861551384486155
0.3 0.841451585484148
0.35 0.810371896281035
0.4 0.789132108678911
0.45 0.780082199178015
0.5 0.749362506374945
0.55 0.723192768072319
0.6 0.699963000369995
0.65 0.674283257167425
0.7 0.6519934800652
0.75 0.636273637263628
0.8 0.612283877161232
0.85 0.582914170858296
0.9 0.577224227757739
0.95 0.546864531354687
1 0.519504804951949
};
\addlegendentry{$\lambda$-0.2}
\addplot [semithick, color2, mark=, mark size=1.5, mark options={solid}, line width=1.5pt]
table {%
0 1.74835251647483
0.05 1.69052309476905
0.1 1.65178348216518
0.15 1.60811391886081
0.2 1.53351466485335
0.25 1.49914500854991
0.3 1.44777552224478
0.35 1.40131598684013
0.4 1.36281637183627
0.45 1.32711672883271
0.5 1.24285757142429
0.55 1.23395766042339
0.6 1.16965830341696
0.65 1.12532874671252
0.7 1.09017909820901
0.75 1.0440695593044
0.8 1.00549994500056
0.85 0.97180028199718
0.9 0.923730762692371
0.95 0.896451035489646
1 0.851091489085111
};
\addlegendentry{$\lambda$-0.3}

\addplot [semithick, color3, mark=, mark size=1.5, mark options={solid}, line width=1.5pt]
table {%
0 2.97968020319793
0.05 2.88303116968829
0.1 2.74669253307465
0.15 2.61793382066179
0.2 2.5359246407536
0.25 2.4570654293457
0.3 2.36921630783688
0.35 2.20622793772068
0.4 2.12315876841233
0.45 2.02405975940238
0.5 1.97488025119748
0.55 1.87315126848731
0.6 1.78655213447865
0.65 1.71761282387176
0.7 1.62982370176298
0.75 1.58576414235858
0.8 1.49221507784921
0.85 1.43034569654302
0.9 1.38045619543804
0.95 1.31779682203178
1 1.24311756882433
};
\addlegendentry{$\lambda$-0.4}

\addplot [semithick, color4, mark=, mark size=1.5, mark options={solid}, line width=1.5pt]
table {%
0 5.85468145318551
0.05 5.36577634223652
0.1 4.98396016039836
0.15 4.72665273347269
0.2 4.35076649233508
0.25 4.1293387066129
0.3 3.83889161108388
0.35 3.60475395246047
0.4 3.35568644313556
0.45 3.179988200118
0.5 3.05883941160591
0.55 2.86510134898652
0.6 2.68364316356838
0.65 2.55562444375556
0.7 2.41019589804102
0.75 2.26955730442696
0.8 2.17142828571717
0.85 2.04886951130486
0.9 1.95473045269547
0.95 1.85270147298526
1 1.74280257197427
};
\addlegendentry{$\lambda$-0.5}

\addplot [semithick, color0, dashed, mark=, mark size=1.5, mark options={solid}, line width=1.5pt]
table {%
0 0.444187036177204
0.05 0.43404340963423
0.1 0.42395898433874
0.15 0.413933198437679
0.2 0.403965497540997
0.25 0.394055334594082
0.3 0.384202169752842
0.35 0.374405470261397
0.4 0.364664710332299
0.45 0.354979371029231
0.5 0.345348940152115
0.55 0.335772912124579
0.6 0.326250787883713
0.65 0.316782074772078
0.7 0.307366286431896
0.75 0.298002942701375
0.8 0.288691569513129
0.85 0.27943169879462
0.9 0.270222868370593
0.95 0.26106462186746
1 0.251956508619564
};
%\addlegendentry{0.1-Approx}
\addplot [semithick, color1, dashed, mark=, mark size=1.5, mark options={solid}, line width=1.5pt]
table {%
0 1.04195211606285
0.05 1.01398633268276
0.1 0.986418590880145
0.15 0.959239649319581
0.2 0.932440565974855
0.25 0.906012685714698
0.3 0.879947628514465
0.35 0.85423727825667
0.4 0.828873772085846
0.45 0.803849490285498
0.5 0.779157046647064
0.55 0.75478927930281
0.6 0.730739241996385
0.65 0.707000195766501
0.7 0.683565601020758
0.75 0.660429109978113
0.8 0.637584559459838
0.85 0.615025964010093
0.9 0.59274750932839
0.95 0.570743545997341
1 0.549008583490075
};
%\addlegendentry{0.2-Approx}
\addplot [semithick, color2, dashed, mark=, mark size=1.5, mark options={solid}, line width=1.5pt]
table {%
0 1.90559975785274
0.05 1.84266085474011
0.1 1.78145490587246
0.15 1.72190403806516
0.2 1.66393522767176
0.25 1.60747991688394
0.3 1.55247366664484
0.35 1.4988558421047
0.4 1.44656932706247
0.45 1.39556026427907
0.5 1.34577781892897
0.55 1.29717396278511
0.6 1.24970327701692
0.65 1.20332277172812
0.7 1.15799172057555
0.75 1.1136715089977
0.8 1.07032549474503
0.85 1.02791887954747
0.9 0.986418590880144
0.95 0.945793172898935
1 0.906012685714698
};
%\addlegendentry{0.3-Approx}
\addplot [semithick, color3, dashed, mark=, mark size=1.5, mark options={solid}, line width=1.5pt]
table {%
0 3.31392688485478
0.05 3.16744416032064
0.1 3.02861367717593
0.15 2.89678703661734
0.2 2.7713909533086
0.25 2.65191635550701
0.3 2.53790936667037
0.35 2.42896379510094
0.4 2.32471484140267
0.45 2.22483379629135
0.5 2.12902354907583
0.55 2.03701476381941
0.6 1.948562608597
0.65 1.86344394542688
0.7 1.78145490587246
0.75 1.70240879109046
0.8 1.62613424607562
0.85 1.55247366664484
0.9 1.48128180478836
0.95 1.41242454375693
1 1.34577781892897
};
%\addlegendentry{0.4-Approx}
\addplot [semithick, color4, dashed, mark=, mark size=1.5, mark options={solid}, line width=1.5pt]
table {%
0 6.2936420287405
0.05 5.83824567518369
0.1 5.43560894878661
0.15 5.07610146923697
0.2 4.7524144484765
0.25 4.45888226354856
0.3 4.19103277784424
0.35 3.94528019151836
0.4 3.7187097474934
0.45 3.50892338636937
0.5 3.31392688485478
0.55 3.1320458689483
0.6 2.96186232941204
0.65 2.80216595603307
0.7 2.65191635550701
0.75 2.51021337933281
0.8 2.37627357500013
0.85 2.24941131633516
0.9 2.12902354907583
0.95 2.01457735812719
1 1.90559975785274
};
%\addlegendentry{0.5-Approx}
\end{axis}

\end{tikzpicture}}}
\subcaption{Mean number of queries $\bar{R}$.}
\label{fig:avgNumCustomersVspWithFixedLambda}
\end{subfigure}%
\caption{Performance as a function of high rate selection probability $p$ for query arrival rate $\lambda \in \set{0.1, 0.2, 0.3, 0.4, 0.5}$. }
\end{figure}
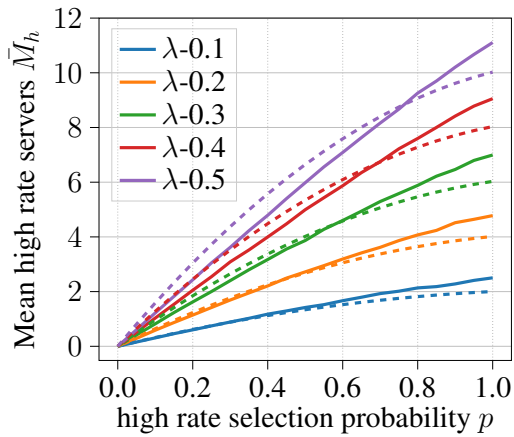
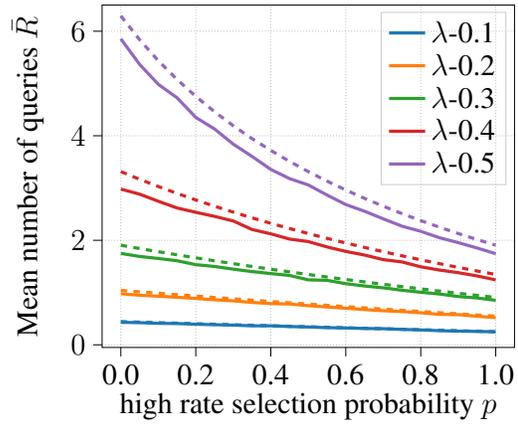
So far, our results were agnostic of the specific form of power consumption as a function of service rate, 
with the only implicit assumption being that the power consumption increases with increase in service rate.  
To study the effect of the system parameters on the power consumption, we need to specify its relation to service rates. 
In practice, the service rate is determined by frequency and voltage, 
and its relationship with power consumption is non-linearly increasing~\cite{MiyoshiCSC2002}. 
We choose a quadratic power consumption function $P(\mu_i) = \alpha \mu_i^2$ with $\alpha = 1$ for a server working at service rate $\mu_i,~i \in \set{0, 1}$.  
\begin{figure}[tbh]
\centerline{\scalebox{0.8}{% This file was created by tikzplotlib v0.9.1.
\begin{tikzpicture}

\definecolor{color0}{rgb}{0.12156862745098,0.466666666666667,0.705882352941177}
\definecolor{color1}{rgb}{1,0.498039215686275,0.0549019607843137}
\definecolor{color2}{rgb}{0.172549019607843,0.627450980392157,0.172549019607843}
\definecolor{color3}{rgb}{0.83921568627451,0.152941176470588,0.156862745098039}
\definecolor{color4}{rgb}{0.580392156862745,0.403921568627451,0.741176470588235}
\definecolor{color5}{rgb}{0.549019607843137,0.337254901960784,0.294117647058824}
\definecolor{color6}{rgb}{0.890196078431372,0.466666666666667,0.76078431372549}
\definecolor{color7}{rgb}{0.737254901960784,0.741176470588235,0.133333333333333}
\definecolor{color8}{rgb}{0.0901960784313725,0.745098039215686,0.811764705882353}

\begin{axis}[
legend cell align={left},
legend style={fill opacity=0.1, draw opacity=1, text opacity=1, draw=white!80!black, at={(0.98,0.42)}},
tick align=outside,
tick pos=left,
x grid style={white!69.01960784313725!black, densely dotted},
xlabel={Mean number of queries $\bar{R}$},
xmajorgrids,
xmin=-0.0503470764311955, xmax=7,
xtick style={color=black},
y grid style={white!69.01960784313725!black, densely dotted},
ylabel={Mean power $\bar{P}$},
ymajorgrids,
ymin=0.0289944716555837, ymax=9.4,
ytick style={color=black},
]
\addplot [semithick, color0, mark=, mark size=1.5, mark options={solid}, line width=1.5pt]
table {%
0.430285697143028 1.18731625883741
0.423995760042398 1.2418964290357
0.411005889941101 1.29163673563264
0.406975930240697 1.35213956660433
0.394066059339408 1.39205920740793
0.384376156238439 1.43999290807093
0.376816231837683 1.49787071729282
0.363626363736366 1.53746700932987
0.363156368436312 1.60569878301218
0.349816501834979 1.63574105858943
0.340946590534096 1.67592605273947
0.328256717432826 1.69786030939689
0.322026779732204 1.76196857631423
0.312806871931278 1.80944252557474
0.308616913830861 1.864489899101
0.294587054129458 1.87685199548004
0.285307146928526 1.92228832911674
0.270217297827021 1.90637944020559
0.261497385026145 1.93911639683599
0.256447435525644 1.99768334716652
0.248597514024861 2.0278552214478
};
\addlegendentry{$\lambda$-0.1}
\addplot [semithick, color1, mark=, mark size=1.5, mark options={solid}, line width=1.5pt]
table {%
0.976960230397696 2.24962484775152
0.949030509694902 2.35937745422546
0.929640703592964 2.48074178058219
0.909790902090979 2.57965410345896
0.886171138288617 2.66860028599715
0.861551384486155 2.76180734992649
0.841451585484148 2.87034723652759
0.810371896281035 2.95265634543655
0.789132108678911 3.03631392886071
0.780082199178015 3.16339193808061
0.749362506374945 3.22770529494702
0.723192768072319 3.29565195148046
0.699963000369995 3.38035960840394
0.674283257167425 3.45051981080189
0.6519934800652 3.50598805611944
0.636273637263628 3.60708346916531
0.612283877161232 3.6722923530764
0.582914170858296 3.70723119168812
0.577224227757739 3.84444155158444
0.546864531354687 3.84958435215641
0.519504804951949 3.87040859591399
};
\addlegendentry{$\lambda$-0.2}
\addplot [semithick, color2, mark=, mark size=1.5, mark options={solid}, line width=1.5pt]
table {%
1.74835251647483 3.24202011979878
1.69052309476905 3.39136415835839
1.65178348216518 3.54915803241971
1.60811391886081 3.70011622283774
1.53351466485335 3.82519808801915
1.49914500854991 3.96115166448339
1.44777552224478 4.11664840151597
1.40131598684013 4.26510922490772
1.36281637183627 4.4012662713373
1.32711672883271 4.53841345186547
1.24285757142429 4.60462154178458
1.23395766042339 4.80104457755423
1.16965830341696 4.88693385066147
1.12532874671252 5.00324026359736
1.09017909820901 5.14046228737711
1.0440695593044 5.22408358316417
1.00549994500056 5.3139097409026
0.97180028199718 5.44409099109007
0.923730762692371 5.50147276127236
0.896451035489646 5.63607397526024
0.851091489085111 5.66582854171457
};
\addlegendentry{$\lambda$-0.3}
\addplot [semithick, color3, mark=, mark size=1.5, mark options={solid}, line width=1.5pt]
table {%
2.97968020319793 4.13105005749942
2.88303116968829 4.33800622793772
2.74669253307465 4.52089203107968
2.61793382066179 4.69647300327
2.5359246407536 4.9033938580614
2.4570654293457 5.08111774482256
2.36921630783688 5.32608035519642
2.20622793772068 5.41772186678135
2.12315876841233 5.59073548864516
2.02405975940238 5.76127831121678
1.97488025119748 5.97429840501594
1.87315126848731 6.11528235517649
1.78655213447865 6.2602990930091
1.71761282387176 6.45293039469606
1.62982370176298 6.57411555084449
1.58576414235858 6.7639238687613
1.49221507784921 6.87059559004409
1.43034569654302 7.01125943140566
1.38045619543804 7.16700596994031
1.31779682203178 7.28397981620173
1.24311756882433 7.33380436195641
};
\addlegendentry{$\lambda$-0.4}
\addplot [semithick, color4, mark=, mark size=1.5, mark options={solid}, line width=1.5pt]
table {%
5.85468145318551 4.96809698303018
5.36577634223652 5.2068521394787
4.98396016039836 5.45794865651344
4.72665273347269 5.68865735342645
4.35076649233508 5.8972850231497
4.1293387066129 6.13870181298188
3.83889161108388 6.32073615263848
3.60475395246047 6.56381342586572
3.35568644313556 6.7456420275798
3.179988200118 6.9890157578425
3.05883941160591 7.21001751182489
2.86510134898652 7.41965443145567
2.68364316356838 7.59103608163913
2.55562444375556 7.78394828851709
2.41019589804102 7.97474706452935
2.26955730442696 8.14107339326602
2.17142828571717 8.38310674893253
2.04886951130486 8.50830752092484
1.95473045269547 8.69673204867963
1.85270147298526 8.85641611983887
1.74280257197427 8.9971146288538
};
\addlegendentry{$\lambda$-0.5}
\addplot [semithick, color0, dashed, mark=, mark size=1.5, mark options={solid}, line width=1.5pt]
table {%
0.444187036177204 0.975476870322856
0.43404340963423 1.04086326104691
0.42395898433874 1.10276970036738
0.413933198437679 1.16119662209334
0.403965497540997 1.21614449992835
0.394055334594082 1.2676138477055
0.384202169752842 1.31560521961976
0.374405470261397 1.36011921045738
0.364664710332299 1.40115645582254
0.354979371029231 1.43871763236111
0.345348940152115 1.47280345798158
0.335772912124579 1.50341469207327
0.326250787883713 1.53055213572157
0.316782074772078 1.55421663192059
0.307366286431896 1.57440906578275
0.298002942701375 1.59113036474585
0.288691569513129 1.60438149877719
0.27943169879462 1.61416348057501
0.270222868370593 1.62047736576708
0.26106462186746 1.62332425310671
0.251956508619564 1.62270528466589
};
%\addlegendentry{0.1-Approx}
\addplot [semithick, color1, dashed, mark=, mark size=1.5, mark options={solid}, line width=1.5pt]
table {%
1.04195211606285 1.95372896174002
1.01398633268276 2.08490049018221
0.986418590880145 2.20911777865727
0.959239649319581 2.32637786375725
0.932440565974855 2.43667801953296
0.906012685714698 2.54001576281136
0.879947628514465 2.63638885845255
0.85423727825667 2.7257953245417
0.828873772085846 2.80823343751184
0.803849490285498 2.88370173719389
0.779157046647064 2.9521990317904
0.75478927930281 3.01372440277042
0.730739241996385 3.06827720968273
0.707000195766501 3.11585709488536
0.683565601020758 3.15646398818932
0.660429109978113 3.1900981114152
0.637584559459838 3.21675998286086
0.615025964010093 3.2364504216794
0.59274750932839 3.24917055216642
0.570743545997341 3.25492180795559
0.549008583490075 3.25370593612236
};
%\addlegendentry{0.2-Approx}
\addplot [semithick, color2, dashed, mark=, mark size=1.5, mark options={solid}, line width=1.5pt]
table {%
1.90559975785274 2.92320435079757
1.84266085474011 3.12013995240384
1.78145490587246 3.30673771783927
1.72190403806516 3.48298381245887
1.66393522767176 3.64886487141018
1.60747991688394 3.80436802391992
1.55247366664484 3.9494809190615
1.4988558421047 4.08419175268209
1.44656932706247 4.20848929521035
1.39556026427907 4.32236292010194
1.34577781892897 4.42580263271211
1.29717396278511 4.51879909941222
1.24970327701692 4.60134367679151
1.20332277172812 4.67342844080683
1.15799172057555 4.73504621576167
1.1136715089977 4.78619060301256
1.07032549474503 4.82685600931487
1.02791887954747 4.85703767473352
0.986418590880144 4.87673170005475
0.945793172898935 4.88593507364531
0.906012685714698 4.88464569771417
};
%\addlegendentry{0.3-Approx}
\addplot [semithick, color3, dashed, mark=, mark size=1.5, mark options={solid}, line width=1.5pt]
table {%
3.31392688485478 3.87650706573059
3.16744416032064 4.13884642408003
3.02861367717593 4.38758982905113
2.89678703661734 4.62270635081384
2.7713909533086 4.84416592321308
2.65191635550701 5.05193922494373
2.53790936667037 5.24599761255006
2.42896379510094 5.42631309454993
2.32471484140267 5.59285833813751
2.22483379629135 5.74560670159695
2.12902354907583 5.88453228687868
2.03701476381941 6.00961000783617
1.948562608597 6.12081567045582
1.86344394542688 6.21812606208464
1.78145490587246 6.30151904720313
1.70240879109046 6.37097366773362
1.62613424607562 6.42647024623558
1.55247366664484 6.46799049063696
1.48128180478836 6.4955175993958
1.41242454375693 6.50903636618981
1.34577781892897 6.50853328340017
};
%\addlegendentry{0.4-Approx}
\addplot [semithick, color4, dashed, mark=, mark size=1.5, mark options={solid}, line width=1.5pt]
table {%
6.2936420287405 4.80267850348873
5.83824567518369 5.13122852769162
5.43560894878661 5.44283420787808
5.07610146923697 5.73746573374737
4.7524144484765 6.01509574169332
4.45888226354856 6.27569755473739
4.19103277784424 6.51924407148549
3.94528019151836 6.74570710106115
3.7187097474934 6.95505701113808
3.50892338636937 7.14726259916656
3.31392688485478 7.32229112415778
3.1320458689483 7.48010845427642
2.96186232941204 7.62067929758735
2.80216595603307 7.74396749171173
2.65191635550701 7.84993633414334
2.51021337933281 7.93854893934551
2.37627357500013 8.00976861199699
2.24941131633516 8.0635592282082
2.12902354907583 8.09988561840949
2.01457735812719 8.11871394707042
1.90559975785274 8.12001208554881
};
%\addlegendentry{0.5-Approx}
\end{axis}
\end{tikzpicture}}}
\label{fig:TradeoffWithFixedLambda}
\caption{Tradeoff between the mean power consumption $\bar{P}$ and the mean number of  queries $\bar{R}$ varying the probability of high rate selection $p$ for fixed arrival rate $\lambda \in \set{0.1, \dots, 0.5}$. 
}
\label{fig:Tradeoff}
\end{figure}
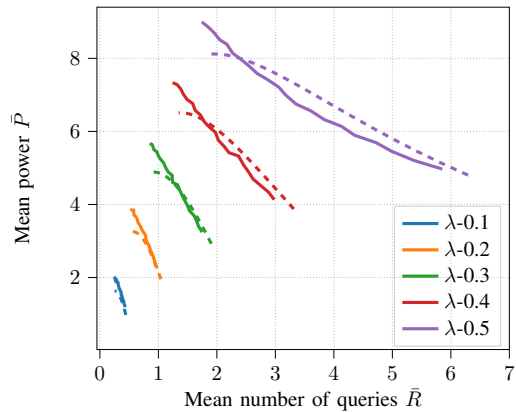
For the quadratic power consumption function, 
we plot the tradeoff curves for mean power consumption with respect to mean number of queries in Fig.~\ref{fig:Tradeoff}. 
Each tradeoff curve is plotted by varying the high rate selection probability $p$ for a fixed arrival rate $\lambda \in \set{0.1, 0.2, 0.3, 0.4, 0.5}$. 
From the tradeoff curve for a fixed arrival rate $\lambda$, 
we observe that probability of high rate selection $p$ can be chosen to minimize the mean power consumption such that the mean sojourn time of queries do not exceed a certain threshold. 

%%%%%%%%%%%%%%%%%%%%%%%%%%%%%%%%%%%%%%%%%
\section{Conclusion \& Future Directions}
\label{sec:conclusion}
%%%%%%%%%%%%%%%%%%%%%%%%%%%%%%%%%%%%%%%%% 
We study a probabilistic server slowdown policy to reduce energy conservation in OLDI applications while meeting a strict delay constraint. 
In such applications, each query is sent to a large number of servers owing to the large size of the data and distributed nature of data storage. 
We have analytically modeled the system evolution of query response systems with independent server slowdowns, 
as a Markov process.
This Markov process remains analytically intractable even in the simple case of single rate servers,  
and hence we provided a theoretically motivated approximation that can be used to compute closed from results for average energy consumption and mean number of queries in the system. 
Under this approximation, we can characterize the energy performance tradeoff and provide guidelines for system parameter selection. 
We demonstrate that this approximation remains tight for all range of underlying parameter values. 

There are multiple interesting future research directions. 
One direction is finding better approximations for multiple rate servers. 
Additionally, there can be arrival of multiple class of queries in the system with varying service requirements where each of these class of queries can be served by different subset of servers. 
For ease of analysis, we have restricted our current  discussions to a system of $n$ servers where all arrivals belong to a single class, which can be serviced by all the available servers.
That is, we have limited the notion of system to the subset of servers which are relevant to the particular class of query that we are interested in. 
Analytical modelling of a more general system incorporating multi-class queries is an interesting direction for future research.

We have chosen a model of server slowdown where each server is slowed down with an \emph{a priori} fixed probability independent of the system state. 
In contrast, an ideal load aware probabilistic rate allocation policy is expected to adjust the probability of rate selection depending on the instantaneous loads. 
Our proposed rate allocation policy is simpler in the sense that it 
%For simplicity, we have assumed that the rate allocation 
depends only on the average load in the system.
%This would require additional system overheads, and hence our proposition is simpler and easier to implement. 
Probabilistic rate allocation is also a practical choice when the total power consumption is influenced by external factors. %like the fluctuations in energy supply.
For instance, if the energy supply fluctuates with respect to time, then the energy expenditure by the data centre can be adapted in order to match the supply curve. %\cite{Klingert2015} 
This can be realized by adapting the rate selection probability at the processor core with respect to the predetermined power supply characteristics.  
Moreover, probabilistic rate control is distributed in nature and can be easily implemented in a practical system with minimum overhead. 
Another interesting extension would be to find a centralized and load aware slowdown for the query response system, where all servers can be simultaneously sped up or slowed down depending on the instantaneous load in the system. 
This extension would also require another set of approximations for performance analysis and system design guidelines.

\bibliographystyle{IEEEtran}
\bibliography{IEEEabrv,performance2021}

\end{document}